\documentclass{article} 
\usepackage[linesnumbered, tworuled, vlined, onelanguage, english]{algorithm2e}
\usepackage[hidelinks]{hyperref}
\usepackage{algpseudocode}
\usepackage{amsmath, amssymb, amsthm} 
\usepackage{amssymb, colordvi} 
\usepackage{authblk}
\usepackage{blkarray} 
\usepackage{centernot}
\usepackage{cleveref} 
\usepackage{enumitem}
\usepackage{fullpage}
\usepackage{graphicx}
\usepackage{mathrsfs} 
\usepackage{multirow} 
\usepackage{mathtools}
\usepackage{stmaryrd}
\usepackage{tikz} 
\usepackage{todonotes}
\usepackage{url}
\usepackage{xcolor} 

\numberwithin{equation}{section} 
\newtheorem{theorem}{Theorem}
\numberwithin{theorem}{section}

\newtheorem{proposition}[theorem]{Proposition}
 
\newtheorem{lemma}[theorem]{Lemma}

\newtheorem{definition}[theorem]{Definition}
\newtheorem{remark}[theorem]{Remark}

\DeclareMathOperator{\Cl}{Cl}
\DeclareMathOperator{\Div}{Div}
\DeclareMathOperator{\Dr}{Dr}

\DeclareMathOperator{\End}{End}

\DeclareMathOperator{\Gal}{Gal}
\DeclareMathOperator{\Hom}{Hom}

\DeclareMathOperator{\Ker}{Ker}
\DeclareMathOperator{\Norm}{Norm}

\DeclareMathOperator{\Pic}{Pic}

\DeclareMathOperator{\lc}{lc}
\DeclareMathOperator{\norm}{\mathfrak{n}}

\DeclareMathOperator{\rgcd}{rgcd}
\DeclareMathOperator{\Frac}{Frac}
\DeclareMathOperator{\val}{val}
\newcommand{\Achar}{\mathfrak{p}}
\newcommand{\A}{\mathbf{A}}

\newcommand{\C}{\mathbf{C}}

\newcommand{\F}{\mathbb{F}}
\newcommand{\Fq}{\mathbb{F}_q}
\newcommand{\HH}{\mathcal{H}}
\newcommand{\Kbar}{{\overline{K}}}
\newcommand{\Ktau}{{K\{\tau\}}}
\newcommand{\K}{\mathbf{K}}
\newcommand{\Lbar}{{\overline{L}}}
\newcommand{\Ltau}{{L\{\tau\}}}
\newcommand{\N}{\mathbb{Z}_{\geqslant 0}}
\newcommand{\OP}{\mathcal{O}_{\Pid}}
\newcommand{\Otilde}{\tilde{\O}}
\newcommand{\Pid}{\mathfrak{P}}
\newcommand{\Q}{\mathbb{Q}}

\newcommand{\Z}{\mathbb{Z}}
\newcommand{\R}{\mathbb{R}}
\newcommand{\aid}{\mathfrak{a}}
\newcommand{\degtau}{\deg_\tau}

\newcommand{\height}{\mathrm{h}}
 
\newcommand{\idef}[1]{\emph{#1}}  

\newcommand{\into}{\hookrightarrow}
\newcommand{\isonorm}{\mathfrak{n}}
\newcommand{\onto}{\twoheadrightarrow}
\newcommand{\pid}{\mathfrak{p}}
\newcommand{\red}{\mathrm{red}}

\newcommand{\DeclareObjectsAsIn}{In this section, $d, m, p, h, f, \xi, \HH,
\A_\HH, \pid, L$ are as in Section~\ref{subsec:correspondence}.}
\renewcommand{\O}{\mathcal{O}}
\renewcommand{\geq}{\geqslant}
\renewcommand{\j}{\mathrm{j}}
\renewcommand{\k}{\mathbf{k}}
\renewcommand{\leq}{\leqslant}
\renewcommand{\ker}{\Ker}
\newcommand{\SMgeq}{\mathrm{SM}^{\geqslant 1}}

\newcommand{\GroupAction}{\textsc{GroupAction}}
\newcommand{\EuclidRGCD}{\textsc{EuclidRGCD}}
\newcommand{\PrimeIsogenyToPrimeIdeal}{\textsc{PrimeIsogenyToPrimeIdeal}}
\newcommand{\IsogenyToIdeal}{\textsc{IsogenyToIdeal}}
\newcommand{\OreEuclideanDivision}{\textsc{OreEuclideanDivision}}

\setlist[enumerate, 1]{label=(\roman*)}
\setlist[enumerate, 2]{label=(\alph*)}
\setlist[enumerate, 3]{label=\arabic*.}
\setlist[enumerate]{font=\normalshape}  
\setlist[itemize, 1]{label=---}
\setlist[itemize, 2]{label=$\ast$}
\setlist[itemize, 3]{label=$\bullet$}
\setlist{noitemsep}

\title{Computing a Group Action from the Class Field Theory of
Imaginary Hyperelliptic Function Fields}

\author{Antoine Leudière and Pierre-Jean Spaenlehauer} 
\affil{Université de Lorraine, Inria, CNRS}

\date{}

\begin{document}
\maketitle

\begin{abstract}
  We explore algorithmic aspects of a simply transitive commutative group
  action coming from the class field theory of
  imaginary hyperelliptic function fields. Namely, the Jacobian of an imaginary
  hyperelliptic curve defined over $\Fq$ acts on a subset of isomorphism
  classes of Drinfeld modules.  We describe an algorithm to compute the group
  action efficiently. This is a function field analog of the
  Couveignes-Rostovtsev-Stolbunov group action. We report on an explicit
  computation done with our proof-of-concept C++/NTL
  implementation; it took a fraction of a second on a standard computer. We prove that the problem of
  inverting the group action reduces to the problem of finding isogenies of
  fixed $\tau$-degree between Drinfeld $\Fq[X]$-modules, which is solvable in
  polynomial time thanks to an algorithm by Wesolowski. We give asymptotic
  complexity bounds for all algorithms presented in this paper.
\end{abstract}

\section*{Introduction}

\paragraph{Context} The class group $\Cl(\Q(\sqrt{-D}))$ of an imaginary
quadratic number field acts on the set of isomorphism classes of elliptic
curves having complex multiplication by $\Q(\sqrt{-D})$. This simply transitive
group action is a central object of the class field theory of imaginary
quadratic number fields, as it provides an explicit way to handle their class
fields and Galois groups. Drinfeld modules --- initially called \emph{elliptic
modules} \cite{drinfel1974elliptic} --- were introduced to create an explicit
and similar class field theory for function fields. There is a comparable group
action in this context; it is expressed in terms of isogenies of Drinfeld
modules. Computing it is the main objective of this work.

Algorithmic questions on Drinfeld modules arose from the start of the
development of the theory, see e.g.~\cite{dummit1994rank}. More recently,
effectivity topics in Drinfeld modules were revisited from the point of view of
modern computer algebra \cite{caranay-article,caranay-thesis}, proposing new
applications, such as factorization of univariate polynomials
\cite{doliskani2021drinfeld} and cryptography \cite{bombar,jn}. In
this paper, we aim at growing the algorithmic toolbox for isogenies of Drinfeld
modules. We focus on their relationship with the class field theory of
imaginary hyperelliptic function fields, and we build on the work of
\cite{caranay-article,caranay-thesis,schost}.

\paragraph{Main results} Although Drinfeld modules might appear more abstract
than elliptic curves, they turn out to be very convenient for concrete
computations. The theory of complex multiplication of Drinfeld modules shares
many similarities with that of elliptic curves. Let $L/\Fq$ be a finite
extension, and let $\phi$ be a rank-two $\Fq[X]$-Drinfeld module defined over
$L$. The characteristic polynomial of the Frobenius endomorphism of $\phi$ has
degree $2$ in $\Fq[X][Y]$, it defines a quadratic extension $\k$ of $\Fq(X)$.
If further $[L:\Fq]$ is odd, $\k$ is an imaginary quadratic function field.
Then, the class field theory of $\k$ gives a simply transitive group action on
the set $S$ of isomorphism classes of Drinfeld modules having complex
multiplication by $\k$. The underlying group $G$ turns out to be the Galois
group of an abelian extension, which is unramified at all finite places and for
which the place at infinity splits completely~\cite[Th.~15.6]{hayes}. For
simplicity, we restrict our work to the case of imaginary hyperelliptic
function fields; $G$ will simply be the degree-$0$ Picard group $\Pic^0(\HH)$
of the underlying hyperelliptic curve $\HH$.

We design an efficient algorithm to compute the group action in the class field
theory of hyperelliptic function field. Surprisingly, the algorithm is quite
easy both to describe and to implement: it mainly relies on computing the
right-GCD of two Ore polynomials. We provide an asymptotic complexity bound. We
also study the difficulty of the so-called \idef{inverse problem}, i.e.
computing an element $g\in G$ such that $g\cdot x = y$, for $x, y\in S$ given.
We prove that it reduces to the problem of computing isogenies of Drinfeld
modules, by providing an algorithm computing the ideal in the coordinate ring
of the hyperelliptic curve which corresponds to a given isogeny. Computing this
isogeny can now be done efficiently using an algorithm by
Wesolowski~\cite{cryptoeprint:2022/438}. We finish our investigation by
providing asymptotic complexity bounds, and we present a concrete calculation
of the group action, using a C++/NTL
implementation.

\paragraph{Related works and applications} An important question is the
computation of zeta functions, which are closely related to the computation of
characteristic polynomials of endomorphisms of Drinfeld modules
\cite[Sec.~5]{gek91}. Isogenies of finite Drinfeld modules are connected to
this line of work by the fact that they correspond to ideals in endomorphism
rings. During the last decade, the algorithmic toolbox gravitating around these
topics and its interaction with computer algebra and algorithmic number theory
has attracted a lot of attention, see
e.g.~\cite{caranay-article,caranay-thesis,schost,doliskani2021drinfeld,
garai2022computing,kuhn2022finding}.

The realization of the class field theory of imaginary quadratic number fields
via isogenies of elliptic curves is the cornerstone of \emph{isogeny-based
cryptography}, whose foundations were laid down by Couveignes
\cite{couveignes2006hard}, and independently by Rostovtsev and Stolbunov
\cite{rostovtsev2006public}. Their ideas paved the way towards the SIDH and the
CSIDH cryptosystems~\cite{sidh,csidh}. In \cite{jn}, the authors propose
Drinfeld modules analogs of SIDH and CSIDH in the supersingular setting, and
they provide polynomial-time attacks on these constructions.

We originally wanted to design an analog of the Couveignes-Rostovtsev-Stolbunov
(CRS) cryptosystem in the context of ordinary Drinfeld modules. In a previous
preprint (\url{https://ia.cr/2022/349}), we replaced elliptic curves by
Drinfeld modules, and the class group of an imaginary quadratic field by the
Jacobian of an hyperelliptic curve. In the meantime, Wesolowski found an
algorithm efficiently computing isogenies of finite Drinfeld $\Fq[X]$-modules
\cite{cryptoeprint:2022/438}. This attack hinders our potential cryptographic
applications. However, this line of research is only flourishing; Drinfeld
modules provide numerous useful features (complex multiplication, efficient
algorithms, deep mathematical theory) and vast areas remain unexplored. For
instance, a recent result showed how Carlitz modules --- which are special
cases of Drinfeld modules --- can be used to design new cryptographic
protocols~\cite{bombar}. This emphasizes the diversity of potential
applications of Drinfeld modules, which may be discovered once a versatile
algorithmic toolbox is available.

\paragraph{Organization of the paper} Section~\ref{sec:drin_mods} recalls the
algebraic construction of Drinfeld modules, as well as basic tools.
Section~\ref{sec:CFT} focuses on complex multiplication and on the class field
theory of imaginary hyperelliptic function fields. The main result of this
section is a reduction of the group action from class field theory to our
specific setting. We can then handle the group action using finite objects. In
Section~\ref{sec:algos}, we describe the main algorithm computing the action.
We also give a method to recover the ideal class corresponding to a given
isogeny of ordinary Drinfeld modules, for which we provide asymptotic
complexity bounds. Finally, we present an explicit computation of
the group action.

\paragraph{Acknowledgements} We are grateful to Xavier Caruso, Pierrick Gaudry, Emmanuel
Thomé, and Benjamin Wesolowski for fruitful discussions. This work received
funding from the France 2030 program managed by the French National Research
Agency under grant agreement No. ANR-22-PETQ-0008 PQ-TLS.

\section{Drinfeld modules}\label{sec:drin_mods}

Classical textbooks on Drinfeld modules are \cite{gos98}, \cite{ros02} and
\cite{villa-salvador}. Finite Drinfeld modules are studied in depth in
\cite{gek91}. For algorithmic perspectives, see \cite{schost},
\cite{caranay-thesis} and \cite{caranay-article}.

Throughout this paper, $\Fq$ is the finite field with $q$ elements.

\subsection{Ore polynomials}

The core mathematical object for the algebraic construction of Drinfeld modules
is the ring $\Ktau$ of univariate Ore polynomials. Let $\Fq \into K$ be a field
extension, and $\tau: x\mapsto x^q$ denote the Frobenius endomorphism of
$\Kbar$, which is $\Fq$-linear.

\begin{definition}[{\cite[Def.~1.1.3]{gos98}}]\label{def:Orepols}
	The \idef{ring of Ore polynomials} $\Ktau$ is the subring of $\Fq$-linear
	endomorphisms of $\Kbar$ of the form \[\sum_{0\leq i\leq n} a_i \tau^i,
	\quad n\in\N, a_i\in K,\] equipped with the addition and the composition of
	$\Fq$-linear endomorphisms.
\end{definition}

In $\Ktau$, we write $1 = \tau^0$ for the identity endomorphism. For $i,j\in\N$,
$\tau^i\tau^j = \tau^{i+j}$. In Definition~\ref{def:Orepols}, and if $a_n\neq 0$, the integer $n$ is called
the $\tau$-degree, and we say that $P$ is \idef{monic} if $a_n = 1$. 
For every $a\in K$, the equality $\tau a = a^q
\tau$ holds true. Therefore, the ring $\Ktau$ is not commutative as soon as
$\Fq\neq K$. The center of $\Ktau$ is $\Fq\left[\tau^{[K:\Fq]}\right]$ if $K$ is
finite, otherwise it is $\Fq$. 

The ring $\Ktau$ is
left-Euclidean~\cite[Prop.~1.6.2]{gos98} for the $\tau$-degree, i.e. for every $P_1, P_2\in\Ktau$,
there exist $Q, R\in\Ktau$ satisfying
\[\begin{cases}
		P_1 = QP_2 + R, \\
		\degtau(R) < \degtau(P_2).
	\end{cases}\]
We therefore define the \idef{right-greatest common divisor}, abbreviated
$\rgcd$, of any non-empty subset $S\subset\Ktau$ as the unique monic generator
of the left-ideal generated by $S$ in $\Ktau$. The $\rgcd$ of two Ore
polynomials can be efficiently computed using
\cite[Alg.~6]{caruso-le_borgne-fact} (altogether with
\cite[Prop.~3.1]{caruso-le_borgne-mult}), or using Euclid's algorithm (see
Algorithm~\ref{algo:ore-euclidean-division}).

We say that $P$ is \idef{separable} if the coefficient of $\tau^0$ is nonzero,
i.e. $\tau$ does not right-divide $P$; we say that $P$ is \idef{inseparable} if
it is not separable; we say that $P$ is \idef{purely inseparable} if
$P=\alpha\tau^i$ for some $\alpha\in K^\times, i\in\Z_{> 0}$. Consequently, for
any $P\in\Ktau$ there exists $\ell\in\N$ and some separable $s\in\Ktau$ such
that $P = \tau^\ell s$. The integer $\ell$ is called the \idef{height} of $P$
and denoted $\height(P)$. Using left-Euclidean division, it can be proved that
for any $P_1, P_2\in\Ktau$ such that $P_1$ is separable, $\Ker(P_2)\subset
\Ker(P_1)$ if and only if $P_2$ right-divides~$P_1$.

\subsection{General Drinfeld modules}\label{subsec:preli::drinfeld-modules}

Let $\k$ be an algebraic function field of transcendence degree $1$ over $\Fq$
(i.e. a finite field extension of $\Fq(X)$), $\infty$ be a place of $\k$, and
$\A\subset \k$ be the ring of functions that are regular outside $\infty$. Let
$K/\Fq$ be a field extension equipped with a $\Fq$-algebra morphism $\gamma:
\A\to K$. The kernel of $\gamma$ is a prime ideal called the
\idef{$\A$-characteristic of $K$}. There are mainly two cases which are of
interest for Drinfeld modules:

\begin{enumerate}
	\item The field $K$ is a finite extension of $\A/\mathfrak p$ for some nonzero prime
		ideal $\mathfrak p\subset \A$, $\gamma$ is the composition $\A\onto\A/\pid\into
		K$, the $\A$-characteristic of $K$ is $\mathfrak{p}$; in this
		case, we will write $L$
    instead of $K$ (see Section~\ref{subsec:preli::finite}).
	\item The field $K$ is a finite extension of $\k$ and the morphism $\gamma$ is injective.
\end{enumerate}

By~\cite[Ch.~7, Cor.~2.7]{lorenzini1996invitation}, quotients of $\A$ by nonzero
ideals $\aid$ are finite-dimensional $\Fq$-vector spaces. The \idef{degree of
$\aid$} is $\deg(\aid)\coloneqq\log_q(\#(\A/\mathfrak a))$. For a nonzero
$a\in\A$, we set $\deg(a) \coloneqq \deg(a\A)$.

\begin{definition}[{\cite[Def. 4.4.2]{gos98}, \cite[Def.~1.1]{gek91}}]\label{d-def}
	A \idef{Drinfeld $\A$-module over $K$} is an $\Fq$-algebra morphism $\phi:
	\A\to\Ktau$ such that,
  \begin{itemize}
		\item for all $a\in \A$, the coefficient of $\tau^0$ in $\phi(a)$ is
			$\gamma(a)$;
		\item there exists $a\in \A$ such that $\deg_\tau(\phi(a))>0$.
  \end{itemize}
\end{definition}

Let $\phi$ be a Drinfeld $\A$-module over $K$. For any $a\in \A$, the image
$\phi(a)$ is denoted $\phi_a$. An important feature of Drinfeld modules is that
there exists an integer $r\in\Z_{>0}$, called \idef{the rank of $\phi$}, such that
$\deg_\tau(\phi_a) = r\deg(a)$ for any $a\in \A$~\cite[Def.~4.5.4]{gos98}. We
let $\Dr_r(\A, K)$ denote the set of Drinfeld $\A$-modules over $K$ with rank
$r$. A special case of interest is when $\A=\Fq[X]$, in which $\phi$ is
uniquely determined by $\phi_X$ and its rank is $\deg_\tau(\phi_X)$.

A Drinfeld $\A$-module $\phi$ induces a $\A$-module law on $\Kbar$, defined by
$(a, x) \mapsto \phi_a(x)$, where $a\in \A, x\in\Kbar$. 
When $\A =\Fq[X]$, this structure of $\Fq[X]$-module on $\Kbar$ can be viewed as
an analog of the $\Z$-module law on the group of points $\mathcal E(\Kbar)$ of
an elliptic curve defined over $K$.

Let $\psi$ be another Drinfeld $\A$-module over $K$. A \idef{morphism of
Drinfeld modules $\iota:\phi\to\psi$} is an Ore polynomial $\iota\in \Ktau$ such that
$\iota\phi_a = \psi_a \iota$ for all $a\in \A$. An \idef{isogeny} is a nonzero morphism.
If $K'/K$ is a field extension, a \idef{$K'$-morphism} (resp.
\idef{$K'$-morphism}) $\phi\to\psi$ is a
morphism that lives in $K'\{\tau\}$. Throughout this paper, isogenies between
$\phi$ and $\psi$ are by default $K$-isogenies. 
The endomorphisms of $\phi$ form a ring
denoted $\End(\phi)$ which always contains $\Fq$ and elements of the form
$\phi_a$, $a\in \A$. Said otherwise, $\A$ is isomorphic to a subring of
$\End(\phi)$. When $\A = \Fq[X]$, endomorphisms $\phi_a$ are analogs of
integer multiplication on elliptic curves.

\subsection{Finite Drinfeld modules}\label{subsec:preli::finite}

In this section, we specialize in Drinfeld $\A$-modu\-les over a finite field
$L$, which are also called \idef{finite Drinfeld modules}. In that case, we fix
$\pid \subset \A$ a nonzero prime ideal and $L$ a finite extension of $\A/\pid$,
equipped with the canonical morphism $\gamma: \A\onto \A/\pid \into L$. Finite
Drinfeld modules have a special endomorphism $\tau_L \coloneqq \tau^{[L:\Fq]}$,
called the \emph{Frobenius endomorphism}. It is worth noticing that for any
Drinfeld $\A$-module $\phi$ over $L$, $\End(\phi)$ contains $\Fq[\tau_L]$. Any
isogeny $\iota:\phi\to\psi$ can be written $\tau^{\ell\deg(\pid)} s$ for some
$\ell\in\N$ and a separable $s\in\Ltau$ \cite[§(1.4), Eq.~(ii)]{gek91}.
Furthermore, the endomorphism $\phi_a$ is separable if and only if $a$ is not
contained in $\Ker(\gamma)=\pid$.

We define now the norm of an isogeny $\iota:\phi\to \psi$ of finite Drinfeld
modules. By
\cite[Th.~III.8.1]{lang}, there exists a map $\chi$, called the
\idef{Euler-Poincaré characteristic}, which sends finite $\A$-modules to ideals
in $\A$ and which satisfies the following properties for any finite $\A$-modules
$M_1, M_2, M_3$: 
\begin{enumerate}
	\item $\chi(0) = \A$ and $\chi(\A/\mathfrak{q}) = \mathfrak{q}$ if
		$\mathfrak{q}$ is prime,
	\item $\chi(M_1) = \chi(M_2)$ if $M_1\simeq M_2$,
	\item $\chi(M_1) = \chi(M_2)\chi(M_3)$ if $0 \to M_2 \to M_1 \to M_3 \to 0$ is a
		short exact~sequence.
\end{enumerate}

The \idef{norm of $\iota$} is defined in \cite[§(3.9)]{gek91} as $\norm(\iota)
\coloneqq \Achar^{\height(\iota)/\deg(\pid)} \chi(\Ker(\iota))$. See
\cite[Lem.~3.10]{gek91} for its properties. 
As for elliptic curves, ``being isogenous'' is an equivalence relation. For
$a\in\A$, we say that $\iota:\phi\to\psi$ is \idef{an $a$-isogeny} if $\iota$
right-divides $\phi_a$ in $\Ltau$. We emphasize that $a$ may not generate the
norm of the isogeny. In fact, $\iota$ is an $a$-isogeny if and only if $a$
belongs to the norm of $\iota$, which need not be a principal ideal. Every isogeny is an $a$-isogeny for some
nonzero $a\in \A$. If $a\notin \pid$ and $\iota$ is an $a$-isogeny, then
$\iota$ is separable
and there exists another separable $a$-isogeny $\hat \iota:\psi\rightarrow\phi$, called the dual $a$-isogeny, such
that $\hat \iota \cdot \iota = \phi_a$ and $\iota\cdot\hat \iota = \psi_a$. See
e.g.~\cite[§(4.1)]{deligne1987survey}.

Drinfeld modules have an analog for Vélu's formula. Let $\iota\in\Ltau$ be
nonzero. There exists a finite Drinfeld $\A$-module $\psi$ defined over $L$
such that $\iota$ is an isogeny $\phi\to\psi$ if and only if $\Ker \iota$ is an
$\A$-submodule of $\Lbar$ (endowed with the $\A$-module structure $(a,
x)\mapsto \phi_a(x)$ for $a\in \A, x\in\Lbar$) and $\deg(\pid)$ divides
$\height(\iota)$~\cite[§(1.4)]{gek91}. We emphasize that for any $a\in \A$, the
Ore polynomial $\psi_a$ can be explicitly computed. The $\tau$-degrees of
$\phi_a$ and $\psi_a$ are equal. By equating the coefficients of $\iota\cdot
\phi_a$ and $\psi_a\cdot \iota$, we obtain simple formulas for computing
iteratively the coefficients of $\psi_a$. For instance, if $\iota$ is
separable, by writing

\[\begin{cases}
	\iota = \sum_{0\leq i\leq \deg_\tau(\iota)} \iota_i \tau^i,\\
	\phi_a = \sum_{0\leq i\leq \deg_\tau(\phi_a)} \lambda_i \tau^i,\\
	\psi_a = \sum_{0\leq i \leq \deg_\tau(\phi_a)} \mu_i \tau^i,
\end{cases}\]
we obtain the following formulas for $i\in\llbracket 0,
\deg_\tau(\phi_a)\rrbracket$:
\begin{equation}\label{eq:Velu}
  \mu_i = \frac{1}{\iota_0^{q^i}}\left(\sum_{0\leq
j\leq i}\iota_j\lambda_{i-j}^{q^j} - \sum_{0\leq j\leq i-1} \mu_j
\iota_{i-j}^{q^j}\right).
\end{equation}

\section{Class field theory of imaginary hyperelliptic function fields}\label{sec:CFT}

\subsection{Complex multiplication for rank-two finite Drinfeld
modules}\label{subsec:class::complexmult-rank2}

Rank-two Drinfeld modules over finite fields enjoy a theory of complex
multiplication which shares many similarities with that of elliptic curves defined over
finite fields. The main difference is that imaginary quadratic number fields are
replaced by imaginary quadratic function fields, namely quadratic extensions of
$\Fq(X)$ for which the place at infinity associated to the discrete valuation
ring \[\{f/g \mid f, g\in\Fq[X],\;
\deg(g)\geq \deg(f)\}\subset \Fq(X)\] ramifies.

We start by fixing a nonzero prime ideal $\pid\subset\Fq[X]$ and by considering
Drinfeld modules in $\Dr_2(\Fq[X], L)$, where $L$ is a finite extension of
$\Fq[X]/\pid$ endowed with the canonical map $\gamma: \Fq[X]\onto\Fq[X]/\pid\into
L$. A Drinfeld module $\phi\in\Dr_2(\Fq[X], L)$ is completely described by the
image $\phi_X$ of $X$, which is an Ore polynomial of the form \[\phi_X = \Delta
\tau^2 + g\tau + \gamma(X), \quad g\in L, \Delta\in L^\times.\]

The Frobenius endomorphism $\tau_L\in \End(\phi)$ satisfies a quadratic equation
\cite[Cor.~3.4]{gek91} \cite[Th.~1]{schost}; for any $\phi\in\Dr_2(\Fq[X], L)$,
there is a unique monic polynomial $\xi$ in $\Fq[X][Y]$ of the
form \[\xi = Y^2 + h(X)Y - f(X)\in\Fq[X][Y],\] such that 
\begin{equation}\label{eq:char_pol}
  \begin{cases}
	\xi(\phi_X, \tau_L) = 0, \\
	\deg(f) = [L:\Fq], \\
  \deg(h) \leq [L:\Fq]\big/2.
\end{cases}
\end{equation}
The polynomial $\xi$ is called the \idef{characteristic polynomial of the
Frobenius endomorphism}. 

\begin{definition}[{\cite[Lemma~(5.2) and Satz~(5.3)]{gekeler1983}}]
\label{def:supersingular}
  Let $\phi\in \Dr_2(\Fq[X], L)$, and let $\xi = Y^2 + h(X) Y - f(X)$ be
  the characteristic polynomial of its Frobenius endomorphism. 
  Then $\phi$ is called \idef{supersingular} if $h\in \pid$, otherwise $\phi$ is
  called \idef{ordinary}.
\end{definition}

If $[L:\Fq]$ is odd and the affine curve defined by $\xi$ in nonsingular, then the degree bounds in \eqref{eq:char_pol} imply that
$\xi$ defines a imaginary hyperelliptic curve over
$\Fq$~\cite[Def.~14.1]{cohen2005handbook}. As $\xi$ is a polynomial of degree
$2$ in $Y$, the affine curve it defines is nonsingular if and only if its discriminant
with respect to $Y$ is squarefree.
In this case, the endomorphism ring
is maximal and completely described:

\begin{proposition}\label{prop:structure-end}
	Assume $[L:\Fq]$ is odd, let $\phi\in\Dr_2(\Fq[X], L)$ be an ordinary
	rank-2 Drinfeld module, and assume that $\xi$ defines an imaginary
	hyperelliptic curve $\HH$. Then $\End_{\Lbar}(\phi) = \End_L(\phi)$.
	Writing $\A_\HH = \Fq[X][Y]/(\xi)$, the
	$\Fq$-algebras $\End_L(\phi)$ and $\A_\HH$ are isomorphic via
  \[\begin{array}{rcl}
		\A_\HH&\to&\End_L(\phi)\\
    \overline{X}&\mapsto&\phi_X \\
    \overline{Y}&\mapsto&\tau_L.
  \end{array}\]
\end{proposition}

\begin{proof}
  By~\cite[Lem.~3.3]{gek91}, the minimal polynomial of $\tau_L$ over $\Fq[X]$ is
  $\xi$, which implies that the kernel of the map
  \[\begin{array}{rcl}
  \Fq[X][Y]&\to&\End_L(\phi)\\
    X&\mapsto&\phi_X\\
    Y&\mapsto&\tau_L
\end{array}\]
  is the ideal generated by $\xi$ and hence $\Fq[\phi_X, \tau_L]$ is isomorphic
  to $\A_\HH$.

  By \cite[Chap.~5, Th.~10.8]{lorenzini1996invitation}, $\A_\HH$ is the
  integral closure of $\Fq[X]$ in the function field $\Fq(\HH)$. Let $\O$ be an $\Fq[X]$-order in
  $\Fq(\HH)$. Since the canonical field extension $\Fq(X)\into
  \Fq(\HH)=\Frac(\A_\HH)$ has degree $2$, $\O$ must be a rank-$2$
  $\Fq[X]$-module. Let $1,\alpha\in\Fq(\HH)$ be an $\Fq[X]$-basis of $\O$. Then
  $\alpha^2 = a + b\alpha$ for some $a,b\in\Fq[X]$, which implies that $\alpha$
  belongs to the $\Fq[X]$-integral closure of $\Fq[X]$ in $\Fq(\HH)$, which is
  $\A_\HH$. This implies that $\O\subset\A_\HH$. Hence, $\A_\HH$ is maximal.

  Since $\tau_L$ is not in the image of the map $g\mapsto \phi_g$, $\Fq[\phi_X,
  \tau_L]\subset \End_L(\phi)$ is a $2$-dimensional $\Fq[X]$-module in
  $\End_L(\phi)\otimes \Fq(X)$. By \cite[Th.~6.4.2.(iii)]{caranay-thesis},
  $\End_L(\phi)\otimes \Fq(X)$ is an imaginary quadratic function field and
  $\End_L(\phi)$ is an $\Fq[X]$-order in it. Therefore, $\Fq[\phi_X, \tau_L]$ is
  an $\Fq[X]$-order in $\End_L(\phi)\otimes \Fq(X)\simeq
  (\Fq[X][Y]/\xi)\otimes\Fq(X)\simeq \Fq(\HH)$. Finally, notice that
  $\End_L(\phi)$ contains the maximal order $\Fq[\phi_X, \tau_L]$, so it must be
  equal to it. 
  It remains to prove that $\End_{\Lbar}(\phi) = \End_L(\phi)$. For any
  finite extension $L'$ of $L$,
  by \cite[Th.~6.4.2.(iii)]{caranay-thesis}, $\End_{L}(\phi)$ is a
  sub-order of $\End_{L'}(\phi)$. As $\End_{L}(\phi)$ is maximal,
  $\End_{L}(\phi)=\End_{L'}(\phi)$.
\end{proof}

The \idef{$\j$-invariant of $\phi$}, denoted $\j(\phi)$, is the quantity
$g^{q+1}/\Delta$~\cite[Def.~5.4.1]{caranay-thesis}. For every $j\in L$,
there exists a Drinfeld module $\phi\in\Dr_2(\Fq[X], L)$ whose $\j$-invariant is
$j$; it is defined by $j^{-1}\tau^2+\tau+\gamma(X)$ if $j\neq 0$, and $\tau^2 +
\gamma(X)$ otherwise. The $\j$-invariant and the characteristic polynomial serve
as classifying criterion \cite[Rem.~5.4.2]{caranay-thesis},
\cite[Th.~3.5]{gek91}; two Drinfeld modules $\phi,\psi\in\Dr_2(\Fq[X], L)$ are:
\begin{enumerate}
	\item $\Lbar$-isomorphic if and only if they have the same $\j$-invariant,
	\item $L$-isogenous if and only if they have the same characteristic
		polynomial.
\end{enumerate}

The next proposition is an analog for Drinfeld modules
of a classical property of endomorphism rings of ordinary elliptic curves
defined over finite fields~\cite[Prop.~4.19]{cox_primes_2013}, see also
\cite[Thm.~3.3]{kuehnert_isomorphic_2022}.

\begin{proposition}\label{prop:Lisog_Lbarisom} 
  Two ordinary Drinfeld modules in
  $\Dr_2(\Fq[X], L)$ are $L$-isomor\-phic if and only if they are $L$-isogenous
  and $\Lbar$-isomorphic.  
\end{proposition}

\begin{proof}
	Let $\phi,\psi\in\Dr_2(\Fq[X], L)$ be two ordinary Drinfeld modules which are
	$L$-isogenous and $\Lbar$-isomorphic. Let $\lambda:\phi\to\psi$ be an $\Lbar$-isomorphism and
	$\iota:\phi\to\psi$ be an	$L$-isogeny, then $\lambda^{-1}\iota\in\End_{\Lbar}(\phi)$. By
	Proposition~\ref{prop:structure-end}, $\End_{\Lbar}(\phi) = \End_L(\phi)$, so
	$\lambda^{-1}\iota\in\Ltau$, and therefore $\lambda\in L$.

	Reciprocally, an $L$-isomorphism is an $L$-isogeny and an $\Lbar$-isomorphism.
\end{proof}

We recall that throughout this paper and unless stated otherwise, isogenies are $L$-isogenies and
Drinfeld modules are called isogenous if they are $L$-isogenous.

\subsection{Rank-one Drinfeld modules on imaginary hyperelliptic
curves}\label{subsec:correspondence}

Let $d\geq 5$ be an odd integer and let $m$ be a positive divisor of $d$. Let
$p\in\Fq[X]$ be a monic irreducible polynomial of degree $d/m$ and let $f =
\alpha p(X)^m\in\Fq[X]$ for some $\alpha\in\Fq^\times$. Finally, let $h\in
\Fq[X]$ be a nonzero polynomial of degree at most $(d-1)/2$ which is not
divisible by $p$. This assumption on $h$ is especially important as it ensures
that we will encounter only ordinary Drinfeld modules, see
Definition~\ref{def:supersingular}. Fix $\xi = Y^2 + h(X)Y - f(X)$ and assume
that $\xi$ defines an imaginary hyperelliptic curve $\HH$, i.e. the curve is
smooth in the affine plane~\cite[Def.~14.1]{cohen2005handbook}. As in
Proposition~\ref{prop:structure-end}, set $\A_\HH = \Fq[X][Y]/(\xi)$. The ring
$\A_\HH$ is
isomorphic to the ring of functions of $\HH$ regular outside the place at
infinity. Let $\pid$ be the prime ideal $\langle p(X), Y\rangle$, which has
degree $d$. Let $L$ be a degree-$m$ extension of $\Fq[X][Y]/\pid$; notice that
$[L:\Fq]=d$ is odd; this will have several technical consequences. Set $\gamma:
\A_\HH \onto \A_\HH/\pid\simeq \Fq[X]/(p)\into L$.

The aim of this section is to prove the following correspondence:
\begin{proposition}\label{prop:correspondence}
  There is a bijection between the set of $\Lbar$-isomorphism
  classes in $\Dr_2(\Fq[X], L)$ containing a
  representative whose characteristic polynomial of the Frobenius endomorphism
  is $\xi$, and the set of $\Lbar$-isomorphism classes
  in $\Dr_1(\A_\HH, L)$.

  This bijection sends the class of a Drinfeld module $\phi\in\Dr_2(\Fq[X], L)$ whose
  characteristic polynomial of the Frobenius endomorphism is $\xi$ to the
  class of $\psi\in\Dr_1(\A_\HH, L)$ where $\psi_{\overline X} =
  \phi_X$ and $\psi_{\overline Y} = \tau_L$.
\end{proposition}

The proof of Proposition~\ref{prop:correspondence} is postponed to the end of this section. 

By using Proposition~\ref{prop:correspondence}, we can define the
\idef{$\j$-invariant} of a Drinfeld module in $\Dr_1(\A_\HH, L)$ as the
$\j$-invariant of its corresponding $\Lbar$-isomorphism class in $\Dr_2(\Fq[X],
L)$.

\begin{lemma}\label{lem:rk1_structure}
  Any rank-$1$ Drinfeld module $\phi\in\Dr_1(\A_\HH, L)$ has the following
  form:
  \[\begin{cases}
    \phi_{\overline X} = \Delta \tau^2 + g\tau +\gamma(\overline X)\\
    \phi_{\overline Y} = \beta\tau_L,
  \end{cases}\]
	where $\Delta\in L^\times$, $g\in L$, $\beta\in\Fq^\times$. Moreover, $\beta$
	is a nonzero square root of $\alpha \Norm_{L/\Fq}(\Delta)$ and it is
	uniquely determined by $\Delta$ and $g$.
\end{lemma}

\begin{proof}
	Since $\overline X$ has degree $2$ in $\A_\HH$ and $\phi$ has rank $1$,
	$\phi_{\overline X}$ must be an Ore polynomial of $\tau$-degree $2$.
	Therefore, $\phi_{\overline X} = \Delta \tau^2 + g\tau +\gamma(\overline X)$
	for some $\Delta\in L^\times$, $g\in L$.

	Next, we show that $\phi_{\overline Y} = \beta \tau_L$ for some
	$\beta\in\Fq^\times$. We start by noticing that since $\phi$ has rank $1$ and
	$\overline Y$ has degree $d$, we must have $\deg_\tau(\phi_{\overline Y}) =
	d$. As $\phi_{\overline{p}}$ has
	constant coefficient zero, \cite[(1.4), Eq.~(ii)]{gek91} implies that
	$\tau^{d/m}$ right-divides $\phi_{\overline{p}}$. Therefore
	$\phi_{\overline f} = \alpha \phi_{\overline p}^m$ is right-divisible by
	$\tau^d = \tau_L$. Since $\overline f$ has degree $2d$ and $\phi$ has rank
	$1$, this implies that $\phi_{\overline f} = w \tau^d$ for some $w\in \Ltau$
	of $\tau$-degree $d$, and consequently $\phi_{\overline Y}\phi_{\overline Y +
	\overline h} = \phi_{\overline f} = w \tau_L$. Since $h$ is not divisible by
	$p$, $\overline Y+\overline h\notin\pid$ and therefore $\phi_{\overline Y +
	\overline h}$ is separable. Consequently, $\phi_{\overline Y + \overline h} =
	w/\beta$ for some $\beta\in L^\times$ and $\phi_{\overline Y} = \beta\tau_L$.

	By examinating the coefficient of $\tau^{2d}$ in the equation
	$\phi_{\overline Y}^2 + \phi_{\overline Y}\phi_{\overline h} =\phi_{\overline
	f}$ we obtain that $\beta^2 = \alpha \Norm_{L/\Fq}(\Delta)$ (as $[L:\Fq]$ is
	odd, $\tau^2$ is a generator of $\Gal(L/\Fq)$). Since $d$ is odd, there is no
	subfield of $L$ of degree $2$ over $\Fq$, and hence $\beta\in\Fq^\times$. We
	then prove that only one square root $\beta$ of $\alpha \Norm_{L/\Fq}(\Delta)$
	is suitable. If $q$ is a power of $2$, then there is only one square root.
	Therefore, let us assume now that $q$ is odd, and let $\pm \delta$ be the two
	distinct square roots of $\alpha \Norm_{L/\Fq}(\Delta)$. By contradiction,
	assume that there exists Drinfeld modules $\psi,\psi'\in\Dr_1(\A_\HH, L)$ such
	that $\psi_{\overline X} = \psi'_{\overline X} = \Delta \tau^2 + g +
	\gamma(\overline X)$ and $\psi_{\overline Y} = \delta\tau_L$, $\psi_{\overline
	Y} = -\delta\tau_L$. Then $0 = \psi_{\overline Y^2 + \overline h \overline
	Y-\overline f} - \psi'_{\overline Y^2 + \overline h \overline Y-\overline f} =
	2\delta \psi'_h \tau_L = 0$, which contradicts the fact that $h\ne 0$.
\end{proof}

\begin{lemma}\label{lem:corresp_surj}
	Any $\phi\in\Dr_1(\A_\HH, L)$ is $\Lbar$-isomorphic to a Drinfeld module
	$\psi\in\Dr_1(\A_\HH, L)$ such that $\psi_{\overline Y} = \tau_L$.
\end{lemma}

\begin{proof}
	By Lemma~\ref{lem:rk1_structure}, $\phi_{\overline Y} = \beta \tau_L$ for some
	$\beta\in\Fq^\times$. Let $\mu\in L^\times$ be an element such that
	$\Norm_{L/\Fq}(\mu)=\beta$ and let $\lambda\in\Lbar^\times$ be a
	$(q-1)$th-root of $\mu$. Then $\lambda^{q^d-1} =
	(\lambda^{q-1})^{1+q+q^2+\dots+q^{d-1}} = \Norm_{L/\Fq}(\mu) = \beta$.
	Direct computations show that the Drinfeld module $\psi\in\Dr_1(\A_\HH, L)$
	defined for all $a\in \A_\HH$ by $\psi_a = \mu\phi_a\mu^{-1}$ satisfies the
	desired property.
\end{proof}

\begin{proof}[Proof of Proposition~\ref{prop:correspondence}]
	To a Drinfeld module $\phi\in\Dr_2(\Fq[X], L)$ with characteristic polynomial
	$\xi$, we associate a Drinfeld module $\psi\in\Dr_1(\A_\HH, L)$ defined by
	$\psi_{\overline X} = \phi_X$ and $\psi_{\overline Y} = \tau_L$.

	Let $\phi' = \alpha\phi\alpha^{-1}\in\Dr_2(\Fq[X], L)$, $\alpha\in\Lbar$, be a
	Drinfeld module $\Lbar$-isomorphic to $\phi$. Note that the characteristic polynomial of
	the Frobenius endomorphism of $\phi'$ need not be $\xi$. We prove that $\psi'$
	defined by $\psi'_{\overline X} = \phi'_{X}$ and $\psi'_{\overline Y} =
	\alpha\tau_L\alpha^{-1}= \alpha^{1-q^{d}}\tau_L$ is a Drinfeld module
  in $\Dr_1(\A_\HH, L)$. Writing $\phi_X = \Delta\tau^2 + g\tau+\gamma(X)$, we must
	have $g\ne 0$ since otherwise $\phi$ would have $\j$-invariant $0$; $\phi$
	would be supersingular~\cite[Lem.~3.2]{bae1992singular}, which contradicts our
	assumption that $h$ is not divisible by $p$ (see
	Definition~\ref{def:supersingular}). Since the coefficient of $\tau$ in
	$\phi'_{X}$ equals $\alpha^{q-1} g$ and is in $L$, we obtain that
	$\alpha^{q-1}\in L$. Then $\alpha^{1-q^{d}}\in L$ as a power of
	$\alpha^{q-1}\in L$. Therefore, $\psi'\in\Dr_1(\A_\HH,L)$. Notice that if
  $\xi(\phi'_X, \tau_L) = 0$, then
  $\alpha\in L$ (Proposition~\ref{prop:Lisog_Lbarisom}), so that
	$\alpha\tau_L\alpha^{-1} = \tau_L$.
	The Drinfeld modules $\psi$ and $\psi'$ are $\Lbar$-isomorphic, and
	we extend our association to a well-defined map from the set
	$\Lbar$-isomorphism classes of Drinfeld modules in $\Dr_2(\Fq[X], L)$
	containing a representative whose characteristic polynomial of the
	Frobenius endomorphism is $\xi$, to the set of $\Lbar$-isomorphism classes of
	Drinfeld modules in $\Dr_1(\A_\HH, L)$. It remains to prove that this map is
	bijective. Injectivity comes easily and surjectivity is a direct consequence
	of Lemma~\ref{lem:corresp_surj}.
\end{proof}

\subsection{A group action from class field theory}\label{subsec:class::action}

Our object of study is a group action of $\Cl(\A)$ on the set of
isomorphism classes of $\Dr_r(\A, K)$, where $\A$ and $K$ are as in
Section~\ref{subsec:preli::drinfeld-modules}. Indeed, if $\mathfrak
a\subset\A$ is a nonzero ideal, we define \[\iota_\aid = \rgcd\left(\{\phi_f : f \in
\mathfrak a\}\right).\] By \cite[Sec.~4]{hayes}, $\iota_\aid$ is a well-defined
isogeny from $\phi$ to some Drinfeld module  denoted
$\aid\star_K \phi\in\Dr_r(\A, K)$. 
This map actually has multiplicative
properties, and principal ideals lead to isogenies that are endomorphisms.
Therefore this map can be extended to a group action of $\Cl(\A)$ on the set of
$K$-isomorphism classes of Drinfeld modules in $\Dr_r(\A, K)$. 
A similar group action in fact appears to be one of the main motivations in the
landmark paper by Drinfeld for making explicit the class field theory of
function fields, see~\cite[Th.~1]{drinfel1974elliptic}.

\smallskip

\DeclareObjectsAsIn

\begin{theorem}\label{thm:grp-action}
	If $\Dr_1(\A_\HH, L)$ is nonempty, then the set of $\Lbar$-isomorphism classes
	of Drinfeld modules in $\Dr_1(\A_\HH, L)$ is a principal homogeneous space for
	$\Cl(\A_\HH)$ under the $\star_L$ action.
\end{theorem}

The proof of Theorem~\ref{thm:grp-action} is postponed to the end of this section.

Theorem~\ref{thm:grp-action} can be seen as a reduction modulo prime ideals of the following
general theorem, which might itself be seen as an function field analog of \cite[Prop.~2.4,
Lem.~2.5.1]{silverman1994advanced}. We emphasize that \cite[Th.~9.3]{hayes} holds in greater generality than what we need here; it holds for any
function field and it is not restricted to hyperelliptic curves.

\begin{theorem}[{\cite[Th.~9.3]{hayes}}]\label{thm:classfieldFF}
	Let $\k$ be the function field of $\HH$, $\K$ be the completion of $\k$ at the
	place $\infty$, and $\C$ be the completion of an algebraic closure $\overline
	\K$. Then the set of $\C$-isomorphism classes of Drinfeld modules in
	$\Dr_1(\A_\HH, \C)$ is a principal homogeneous space for $\Cl(\A_\HH)$ under
	the $\star_{\C}$ action.
\end{theorem}

Our strategy to prove Theorem~\ref{thm:grp-action} is to use the reduction and
lifting properties of ordinary Drinfeld modules
\cite[Sec.~11]{hayes}\cite[Th.~3.4]{bae1992singular}.

Let $K$ be a finite extension of $\k$. Let $\Pid$ be a place of $K$ above
$\pid\subset \A_\HH$ and $\OP$ be the associated discrete valuation ring, with
the associated reduction morphism $\red_\Pid:\OP\onto \OP/\Pid$. An Ore
polynomial $f\in K\{\tau\}$ is said to be \idef{defined over $\OP$} if its
coefficients lie in $\OP$ and its leading coefficient is invertible in $\OP$. A
Drinfeld $\A_\HH$-module $\phi$ over $K$ is said to be \idef{defined over $\OP$}
if for all $a\in \A_\HH$,  $\phi_a$ is defined over $\OP$. Let
$\Dr_{r,\Pid}(\A_\HH, K)$ be the set of Drinfeld modules defined over $\OP$. By
considering the morphism $\gamma:\A_\HH\onto\A_\HH/\pid \into \OP/\Pid$, the
reduction map $\red_\Pid$ extends canonically to a map $\Dr_{r,\Pid}(\A_\HH,
K)\to\Dr_{r}(\A_\HH, \OP/\Pid)$.
\begin{lemma}\label{lem:Pidstab}
	For any $\phi\in\Dr_{r, \Pid}(\A_\HH, K)$ and any ideal $\mathfrak a\subset
	\A_\HH$, the Drinfeld module $\aid\star_K\phi$ is defined over $\OP$ and
	\[\red_\Pid(\mathfrak a\,\star_K\, \phi) = \mathfrak a\,\star_{(\OP/\Pid)}\,
	\red_\Pid(\phi).\]
\end{lemma}

\begin{proof}
	The Drinfeld module $\mathfrak a\,\star_K\, \phi$ is defined over $\OP$ by
	\cite[Prop.~11.2]{hayes}, hence $\red_\Pid(\mathfrak a\,\star_K\,\phi)$ is
	well-defined. Let $\iota_{\mathfrak a}$ be the monic generator of the left-ideal in $\Ktau$ generated by
	$\{\phi_g : g\in \mathfrak a\}$.
	Since $\phi$ is defined over $\OP$, we deduce that $\iota_{\mathfrak a}$ must have
	coefficients in $\OP$ and that its reduction generates the left-ideal in
	$(\OP/\Pid)\{\tau\}$ generated by $\{\red_{\Pid}(\phi_g) : g\in \mathfrak
	a\}$. Consequently, $\red_\Pid (\iota_{\mathfrak a})$ is the isogeny associated to
	$\mathfrak a\,\star_{(\OP/\Pid)}\, \red_\Pid(\phi)$, which concludes the
	proof.
\end{proof}

\begin{proof}[Proof of Theorem~\ref{thm:grp-action}]

  Using the correspondence of Proposition~\ref{prop:correspondence} between
  $\Lbar$-isomor\-phism classes of Drinfeld modules, we can associate to any
  Drinfeld module in $\Dr_1(\A_\HH, L)$ a Drinfeld module in $\Dr_2(\Fq[X], L)$
  whose characteristic polynomial of the Frobenius endomorphism is $\xi$.
  Throughout this proof, we fix a place $\Pid$ of $\overline{\Fq(X)}$ above
  $\pid$. Such a place defines a compatible discrete valuation ring $\OP^{(K)}$
  in any finite extension $K$ of $\Fq(X)$.

	Let us prove the transitivity of the action. Let $j_1, j_2\in L$ be the
	$\j$-invariants of two Drinfeld modules $\phi,\psi\in \Dr_2(\Fq[X], L)$, whose
	characteristic polynomial of the Frobenius is $\xi$. Since the ideal $\langle
	p(\overline X)\rangle$ splits in $\A_\HH$ ($\langle p(\overline
	X)\rangle=\langle p(\overline X), \overline Y\rangle\cdot \langle p(\overline
	X), \overline Y+h(\overline X)\rangle$), Deuring's lifting theorems for
	Drinfeld modules \cite[Th.~3.4, Th.~3.5]{bae1992singular}
	(see~\cite[Ch.~13, §4]{lang1987elliptic} for the analogs for elliptic curves) imply that there
	exists a finite extension $K$ of $\Fq(X)$ and two $\C$-isomorphism classes of
	Drinfeld modules in $\Dr_2(\Fq[X], K)$, whose $\j$-invariants reduce to $j_1,
	j_2$ modulo $\Pid$. Those $\j$-invariants are algebraic integers in $\C$
	\cite[§(4.3)]{gekeler1983}. Moreover, those classes contain Drinfeld modules
	$\phi', \psi'\in\Dr_2(\Fq[X], K)$ whose endomorphism rings are isomorphic to
	$\End(\phi) \simeq\End(\psi)\simeq\A_\HH$. Therefore those Drinfeld modules
	can be regarded as Drinfeld modules in $\Dr_{1,\Pid}(\A_\HH, K)$. Since
	$\star_K$ acts on $\Dr_{1, \Pid}(\A_\HH, K)$ \cite[Prop.~11.2]{hayes}, and
	the group action associated to $\star_\C$ is transitive
	(Theorem~\ref{thm:classfieldFF}), there is an ideal $\aid\subset\A_\HH$ such
	that $\aid\star_K \phi'$ is isomorphic to $\psi'$. Consequently, the
	$\j$-invariants $\aid\star_K \phi'$ and $\psi'$ are equal, and therefore their
	reduction modulo $\Pid$ equals $j_2$. Using Lemma~\ref{lem:Pidstab}, the
	$\j$-invariant of $\aid\star_K \phi'$ reduces modulo $\Pid$ to the
	$\j$-invariant of $\aid\star_{\OP^{(K)}/\Pid} \phi$, which therefore also
	equals $j_2$. Hence $\aid$ sends the $\Lbar$-isomorphism class of $\phi$ to
	that of $\psi$ via the $\star_{\Lbar}$ action (which is the same as the
	$\star_L$-action on $\phi$, since $\phi$ is defined over $L$).

	Finally, let us prove the freeness of the action. Let $\phi\in\Dr_1(\A_\HH, L)$
	be a Drinfeld module, and set $\psi = \aid\,\star_L\, \phi$. Assume that
	$\phi$ and $\psi$ are $\Lbar$-isomorphic. Since $\phi$ and $\psi$ are
	$L$-isogenous, by Proposition~\ref{prop:Lisog_Lbarisom} they must be
	$L$-isomorphic. Let $\alpha\in L$ be such an isomorphism, i.e.
	$\alpha\phi\alpha^{-1} = \psi$. Using \cite[Th.~3.4]{bae1992singular} as
	above, the lifting procedure provides us with $\phi'\in\Dr_{1,\Pid}(\A_\HH,
	K)$ which reduces to $\phi$ modulo $\Pid$. Then set $\psi' = \aid\,\star_K\,\phi'$, and let $\iota_{\aid}$ be the
	associated isogeny. By the same argument as in the proof of
	Lemma~\ref{lem:Pidstab}, we obtain that $\iota_\aid$ is defined over
	$\OP^{(K)}$ and that $\red_\Pid(\iota_\aid) = \alpha$ which implies that
	$\iota_{\aid}\in K$, and therefore $\phi'$ and $\psi'$ are isomorphic.
	Consequently, $\aid$ is principal (Theorem~\ref{thm:classfieldFF}), and hence
	the group action associated to $\star_L$ is free.
\end{proof}

\section{Algorithms}\label{sec:algos}

\DeclareObjectsAsIn{} We also fix $\omega \coloneqq \gamma(X) \in L$.

\subsection{Computation of the group action}
\label{subsec:computation-group-action}

Before describing the algorithm for computing the group action in
Theorem~\ref{thm:grp-action}, we need data structures to represent elements in
$\Cl(\A_\HH)$ and $\Lbar$-isomorphism classes. Thanks to
Proposition~\ref{prop:correspondence}, we can use $\j$-invariants --- which are
elements of $L$ --- to represent $\Lbar$-isomorphism classes of Drinfeld modules
in $\Dr_1(\A_\HH, L)$. For representing elements in $\Cl(\A_\HH)$, we use
Mumford coordinates~\cite[Th.~14.5]{cohen2005handbook}: in our case
$\Cl(\A_\HH)$ is isomorphic to $\Pic^0(\HH)$. 

\begin{lemma}\label{lem:pic0-clB}
  The ring $\A_\HH$ is a Dedekind domain, and $\Cl(\A_\HH) \simeq \Pic^0(\HH)$.
\end{lemma}

\begin{proof}
	The ring $\A_\HH$ is a Dedekind domain because $\HH$ is smooth in the affine
	plane~\cite[Ch.~7, Cor.~2.7]{lorenzini1996invitation}. The isomorphism
	$\Cl(\A_\HH) \simeq \Pic^0(\HH)$ comes from the fact that there is a unique
	degree-$1$ place $\infty$ at infinity. Indeed, the group of affine divisors
	$\Div(\A_\HH)$ (i.e. the subgroup of divisors whose valuation at infinity is
	$0$) is isomorphic to the group of degree-$0$ divisors in $\Div_0(\HH)$ via
	the map which sends a divisor $D$ in $\Div(\A_\HH)$ to $D-\deg(D)\infty$.
	Next, we notice that $D$ is principal in $\Div(\A_\HH)$ if and only if its
	image in $\Div_0(\HH)$ is principal. We conclude by using the isomorphism
	in~\cite[Ch.~7, Prop.~7.1]{lorenzini1996invitation}, which shows that the
	quotient of $\Div(\A_\HH)$ by principal divisors is isomorphic to $\Cl(\A_\HH)$.
\end{proof}

Since $\HH$ has genus $\lfloor (d-1)/2\rfloor$, elements in
$\Pic^0(\HH)$ can be represented by \idef{Mumford
coordinates}~\cite[Th.~14.5]{cohen2005handbook}, which are pairs of polynomials $(u, v)\in\Fq[X]^2$
such that:
\begin{enumerate}
	\item $u$ is a nonzero monic polynomial of degree at most $(d-1)/2$,
	\item $\deg(v)<\deg(u)$,
	\item $u$ divides $\xi(X, v(X))$.
\end{enumerate}
Mumford coordinates $(u,v)$ encode the ideal class of
$\langle u(\overline X), \overline Y-v(\overline X)\rangle\subset \A_\HH$.

\begin{algorithm}\label{algo:groupaction}
  \caption{\GroupAction}
  \KwIn{
		\begin{itemize}
			\item A $\j$-invariant $j\in L$ encoding an isomorphism class $\mathcal C$
				of Drinfeld modules in $\Dr_1(\A_\HH, L)$.
			\item Mumford coordinates $(u, v)\in \Fq[X]^2$ for a divisor class $[D]$ in
				$\Pic^0(\mathcal H)$.
		\end{itemize}
  }
  \KwOut{The $\j$-invariant obtained by making $[D]$
  act on $\mathcal C$ by the $\star_L$ action.}
  $\widetilde u \gets u(j^{-1}\tau^2 + \tau + \omega)\in L\{\tau\}$\;\label{step:eval-1}
  $\widetilde v \gets v(j^{-1}\tau^2 + \tau + \omega)\in L\{\tau\}$\;\label{step:eval-2}
  $\iota \gets \rgcd(\widetilde u, \tau_L-\widetilde v)$\label{step:iota}\tcc*{$\iota = \sum_{0\leq k\leq
  \deg_\tau(\iota)}\iota_k\tau^k$}
  $\widehat g \gets \iota_0^{-q}(\iota_0
  +\iota_1(\omega^q-\omega))$\label{step_algo:inverti0}\;
  $\widehat \Delta \gets j^{-q^{\deg_\tau(\iota)}}$\;
  \KwRet ${\widehat g}^{q+1}/{\widehat \Delta}$.
\end{algorithm}

\begin{proposition}
  Algorithm~\ref{algo:groupaction} (\GroupAction) is correct.
\end{proposition}

\begin{proof}
  A representative of the class in $\Cl(\A_\HH)\simeq \Pic^0(\HH)$ represented
  by the Mumford coordinates $(u,v)$ is the ideal $\langle u(\overline X),
  \overline Y-v(\overline X)\rangle\subset \A_\HH$. A representative of the
  isomorphism class of Drinfeld modules represented by the $\j$-invariant $j$
  is a Drinfeld module $\phi\in\Dr_1(\A_\HH, L)$ such that $\phi_{\overline
  X} = j^{-1}\tau^2 + \tau + \omega$ and $\phi_{\overline Y} = \beta\tau_L$ for
  some $\beta\in\Fq^\times$ (see Section~\ref{subsec:correspondence}). Note
  that $j\ne 0$ by \cite[Lem.~3.2]{bae1992singular}. We shall prove that
  $\langle u(\overline X), \overline Y-v(\overline X)\rangle\star_L \phi =
  \psi$, where $\psi\in\Dr_1(\A_\HH, L)$ is the Drinfeld module such that
  $\psi_{\overline X} = \widehat \Delta\tau^2 + \widehat g\tau +\omega$ and
  $\psi_{\overline Y} = \beta\tau_L$. 

  The Ore polynomial $\iota$ computed at Step~\ref{step:iota} is
  $\rgcd(\phi_{u(\overline{X})}, \tau_L-\phi_{v(\overline{X})})$, which is by
  construction the monic Ore polynomial defining the isogeny. Since we need to
  invert the coefficient $\iota_0$ (at Step~\ref{step_algo:inverti0}), we need
  to prove that $\iota$ is separable. This is indeed true: $\iota$
  right-divides $\phi_{u(\overline X)}$, which is separable because $\deg(u) <
  d$. Hence $u$ cannot be a multiple of $p$, which is a generator of
  $\Ker(\gamma)$.

  Since $\iota$ is an isogeny \cite[Cor.~5.10]{hayes}, there exists
  $\psi\in\Dr_1(\A_\HH, L)$ such that $\iota\cdot \phi_{\overline X} =
  \psi_{\overline X}\cdot \iota$ where $\psi_{\overline X}$ has $\tau$-degree
  $2$. It remains to prove that $\psi_{\overline X} = \widehat\Delta\tau^2 +
  \widehat g\tau +\omega$. This is done by extracting as in
  Equations~\eqref{eq:Velu} the coefficients of $\tau$ and
  $\tau^{\deg_\tau(\iota)+2}$ in the equality $\iota\cdot \phi_{\overline X} =
  \psi_{\overline X}\cdot \iota$, which provides us with:

	\[\left\{\begin{array}{rcl}
		\iota_0\,g +\iota_1\,\omega^q &=& \widehat g\,\iota_0^q +
		\omega\,\iota_1,\\
		j^{-q^{\deg_\tau(\iota)}} &=& \widehat \Delta.
	\end{array}\right.\]

  There is only one pair $(\widehat \Delta, \widehat g)\in L^2$ which satisfies
  these two equalities, and the associated Drinfeld module has $\j$-invariant
  ${\widehat g}^{q+1}/{\widehat \Delta}$.
\end{proof}

\begin{algorithm}\label{algo:ore-euclidean-division}
  \caption{\OreEuclideanDivision}
  \KwIn{Two Ore polynomials $a,b \in \Ltau$.}
  \KwOut{Ore polynomials $q, r \in \Ltau$ such that $a = qb + r$ and $\deg_\tau(r) < \deg_\tau(b)$.}
    $q \gets 0$ \;
    $r \gets a$ \;
    \While {$\deg_\tau(r) \geqslant \deg_\tau(b)$}{
      $\varepsilon \gets \lc(r)
       \cdot \tau^{\deg_\tau(r) - \deg_\tau(b)} \cdot (\lc(b))^{-1}$ \;
      $q \gets q + \varepsilon$ \;
      $r \gets r - \varepsilon \cdot b$ \;
    }
    \KwRet $(q, r)$.
\end{algorithm}

\begin{lemma}\label{lemma:correct-euclidean-division}
  
  Algorithm~\ref{algo:ore-euclidean-division} (\OreEuclideanDivision)
  terminates and is correct.
    
\end{lemma}

\begin{proof}

  As the degree of $r$
  decreases at each such call, the algorithm must terminate.
  At each recursive call, the relation $a = qb + r$ holds, which implies that
  the algorithm is correct 

\end{proof}

We finish this section by studying the asymptotic complexity of
Algorithm~\ref{algo:groupaction}. For Ore Euclidean division and right-greatest
common divisor computation, we use Algorithms~\ref{algo:ore-euclidean-division}
and \ref{algo:euclide-rgcd}, mimicking
naïve algorithms for usual (commutative) univariate polynomials. In fact, using
the fastest known algorithms for Ore polynomial multiplication
and Euclidean division in 
\cite{caruso-le_borgne-fact, caruso-le_borgne-mult} would not improve our
complexity bound for Algorithm~\ref{algo:groupaction}, see Remark~\ref{rem:caruso-le_borgne}.

Here and subsequently, if $f$ and $g$
are two functions defined on $\Z_{\geq 0}^2$, with values in $\R_{>0}$, we write $f =
O(g)$ if there exists $M > 0$ such that for every $(x, y) \in \Z_{\geq 0}^2$, $f(x, y) \leqslant M g(x, y)$.
Also, by "application of the Frobenius
endomorphism", we mean computing $\lambda^q$ given $\lambda \in L$. 

\begin{lemma}\label{lem:compl-euclidean-division}
  
  Assuming $\deg_\tau(a) > \deg_\tau(b)$,
  Algorithm~\ref{algo:ore-euclidean-division} (\OreEuclideanDivision) requires
  $O(\deg_\tau(b)(\deg_\tau(a) - \deg_\tau(b))$ arithmetic operations in $L$
  and $O(\deg_\tau(b)(\deg_\tau(a) - \deg_\tau(b))$ applications of the
  Frobenius endomorphism.

\end{lemma}

\begin{proof}

  First, we notice that in the worst-case scenario, the algorithm needs to
  compute $\tau^{\deg_\tau(r) - \deg_\tau(b)} \lc(b)^{-1}$ and
  $\tau^{\deg_\tau(r) - \deg_\tau(b)} \lc(b)^{-1} b$  for $\deg_\tau(r)$ ranging
  from $\deg_\tau(b)$ to $\deg_\tau(a)$. This can be precomputed for
  $O(\deg_\tau(b)(\deg_\tau(a)-\deg_\tau(b))$ operations in $L$ and
    applications of the Frobenius endomorphism. 

    At each step of the loop, computing $\varepsilon$ and
  $q+\varepsilon$ costs a constant
  number of operations, and computing $r$ costs $O(\deg_\tau(b))$ operations.
  In total, this amounts to $\deg_\tau(b)(\deg_\tau(a) - \deg_\tau(b))$ operations
  in $L$ and the same upper bound for the number of applications of the
  Frobenius endormorphism.
\end{proof}

\begin{algorithm}\label{algo:euclide-rgcd}
  \caption{\EuclidRGCD}
  \KwIn{Two Ore polynomials $a=\sum_{0\leq i\leq \deg_\tau(a)} a_i\tau^i,
  b=\sum_{0\leq i\leq \deg_\tau(b)} b_i\tau^i$ in $L\{\tau\}$, such that $a\ne 0$.}
    \KwOut{The right-gcd of $a$ and $b$.}
    \If{$b = 0$}{\KwRet $a$.}
    \If{$\deg_\tau(b) > \deg_\tau(a)$}{\KwRet $\EuclidRGCD(b,a)$.}
    $(q, r) \gets \OreEuclideanDivision(a, b)$ \;
    \KwRet $\EuclidRGCD\left(r, b\right)$.
\end{algorithm}

\begin{lemma}\label{prop:euclideRGCD_correct}
  Algorithm~\ref{algo:euclide-rgcd} (\EuclidRGCD) terminates and is correct.
\end{lemma}

\begin{proof}

  This algorithm is the classical Euclid's algorithm for computing a gcd and
  its proof of correctness is similar to the classical case.

\end{proof}

The following lemma yields a uniform complexity bound for the rgcd in terms of all the
parameters $q, d, \deg_\tau(a), \deg_\tau(b)$.

\begin{lemma}\label{lem:complexite-euclide-rgcd}
  
  Algorithm~\ref{algo:euclide-rgcd} (\EuclidRGCD)
  requires at most $O(\deg_\tau(a)\deg_\tau(b))$ arithmetic operations in $L$ and $O(\deg_\tau(a)\deg_\tau(b))$ applications of
  the Frobenius endomorphism.

\end{lemma}

\begin{proof}
  This complexity is proved by using the standard methods to evaluate the
  complexity of Euclid's algorithm from the complexity of the Euclidean
  division in $O(\deg_\tau(b)(\deg_\tau(a)-\deg_\tau(b)))$ operations
  (Lemma~\ref{lem:compl-euclidean-division}), see
  e.g.~\cite[Th.~3.11]{von_zur_gathen_modern_2013}.
\end{proof}



\begin{proposition}\label{prop:complexity:groupaction}

  Algorithm~\ref{algo:groupaction} requires $O(d^2)$ operations in $L$ and
  $O(d^2)$ applications of the Frobenius endomorphism.

\end{proposition}

\begin{proof}

  Writing $u = u_\ell X^\ell + \dots + u_0$ and $\phi_X =
  j^{-1}\tau^2+\tau+\omega$, we have $\ell\leq(d-1)/2$ and $\widetilde{u} =
  u_\ell \phi_X^\ell + \dots + u_0$. In order to compute $\tilde u$, we can
  first compute $\phi_X^2, \dots, \phi_X^\ell$ iteratively. Let $n\in\llbracket
  1,\ell-1\rrbracket$ and write $\phi_X^n = \sum_{i=0}^{2n} a_i \tau^i$. Then
    \[\phi_X\phi_X^n = \sum_{i=0}^{2n}\left(
      a_i\omega \tau^i
      + g a_i^q\tau^{i+1}
      + \Delta a_i^{q^2} \tau^{i+2}\right).\]

      Knowing $\phi_{X}^n$, the computation of $\phi_X^{n+1}$ requires $O(n)$ additions, multiplications,
  $q$-exponen\-tia\-tions and $q^2$-exponentiations, which is $O(n)$ operations in
  $L$ and $O(n)$ applications of the Frobenius endomorphism of $L/\Fq$.
  Consequently, $O(d^2)$ operations in $L$ and $O(d^2)$ applications of the
  Frobenius endomorphism are required to compute $\phi_X^2, \dots,
  \phi_X^\ell$.

  The last operation that will affect the asymptotic complexity is the rgcd, which we perform using Algorithm~\ref{algo:euclide-rgcd}.
  We have $\deg(u) = \ell$, $\deg(v) < \ell$, so that $\tilde{v}$ and $\tau_L -
  \tilde{u}$ respectively have $\tau$-degree at most $d$. 
  By Lemma~\ref{lem:complexite-euclide-rgcd}, this algorithm
  requires $O(d^2)$ operations in $L$ and $O(d^2)$ applications of the
  Frobenius endomorphism.


\end{proof}
  
\subsection{Computation of the ideal corresponding to an
isogeny}\label{subsec:algo::effective-transitivity}

In this section, we make explicit the transitivity of the group action. Given two
Drinfeld modules $\phi, \psi\in\Dr_1(\A_\HH, L)$, our goal is to compute Mumford coordinates
$(u,v)\in\Fq[X]^2$ such that the class of $\langle u(\overline X),  \overline
Y-v(\overline X)\rangle\subset\A_\HH$ sends the
$\Lbar$-isomorphism class of $\phi$ to that of $\psi$, via $\star_L$.
We emphasize that, given $\phi,\psi\in\Dr_1(\A_\HH, L)$, computing an isogeny
$\iota$ between $\phi$ and $\psi$ can
be achieved efficiently by using Wesolowski's
method~\cite{cryptoeprint:2022/438}. Our algorithm then converts such an
isogeny $\iota$ into the desired Mumford coordinates $(u,v)$. For simplicity, we shall assume
that the norm of the isogeny is coprime to $\mathfrak p$, in order to avoid
separability issues. In the general case, once the part of the isogeny which is
coprime to the characteristic has been treated, the part whose norm is a power
of $\mathfrak p$ can be computed easily since it is either a power of the
Frobenius or a power of its dual, and these cases can be easily discriminated.

We use the shorthand notation $\Dr_2(\Fq[X], L)_\xi$ to denote the subset of
Drinfeld modules in $\Dr_2(\Fq[X], L)$ whose
characteristic polynomial of the Frobenius endomorphism is $\xi$. By
Proposition~\ref{prop:correspondence}, to any $\phi\in\Dr_1(\A_\HH, L)$, we can
associate a Drinfeld module $\phi'\in\Dr_2(\Fq[X], L)_\xi$. Notice that
$\star_L$ leaves $\Dr_2(\Fq[X], L)_\xi$ globally invariant. Hence, by slight
abuse of notation, we shall use the $\star_L$ notation to also denote the
corresponding action of nonzero ideals in $\A_\HH$ over $\Dr_2(\Fq[X], L)_\xi$.
Another useful remark is that computing Mumford coordinates for the class of a
given ideal in $\A_\HH$ can be done efficiently by using the reduction step of
Cantor's algorithm~\cite[Algo.~14.7]{cohen2005handbook}. Therefore, our main
algorithmic task is to construct the ideal in $\A_\HH$ corresponding to a given
isogeny.

We start by the following lemma, which establishes a correspondence between
ideals in $\A_\HH$ and isogenies:

\begin{lemma}\label{lem:corres_id_isog}
  Let $\phi\in\Dr_2(\Fq[X], L)_\xi$ be an ordinary Drinfeld module. Then there
  is a one-to-one correspondence between monic isogenies with domain $\phi$ and
  nonzero ideals in $\A_\HH$. Moreover, let
  $\phi_1,\phi_2,\phi_3\in\Dr_2(\Fq[X], L)_\xi$ be Drinfeld modules and
  $\iota_1:\phi_1\to\phi_2$, $\iota_2:\phi_2\to\phi_3$ be isogenies; the ideal associated to
  $\iota_2\cdot\iota_1$ in $\A_\HH$ is the product of the ideals associated to
  $\iota_1$ and $\iota_2$.
\end{lemma}

\begin{proof}
  To any monic isogeny $\iota:\phi\to\psi$, we
  associate the nonzero ideal $\Hom(\psi, \phi)\iota\subset\End(\phi)\simeq
  \A_\HH$. Notice that $\psi\in\Dr_2(\Fq[X], L)_\xi$ (Section~\ref{sec:CFT}).
  Reciprocally, to any nonzero ideal $\aid\subset\A_\HH$ corresponds the isogeny which is
  the monic generator of the left-ideal in $\Ltau$ generated by $\{g(\phi_X,
  \tau_L) : g\in \aid\}$. We refer to \cite[§(3.6)]{gek91} for more details.

  To prove the second statement, we start by letting $\Xi$ denote the
  isomorphism between $\End(\phi)$ and $\End(\psi)$ which sends $g(\phi_X,
  \tau_L)$ to $g(\psi_X, \tau_L)$ for any $g\in\A_\HH$. Let $\widehat\iota$ be a
  $u$-dual isogeny for $\iota$, for some $u\in\Fq[X]$ such that $\iota$
  right-divides $\phi_u$ (see Section~\ref{subsec:preli::finite}). Notice that for
  all $g\in\A_\HH$, $\phi_u$ right-divides $g(\phi_X, \tau_L)\cdot\widehat\iota$
  if and only if $\psi_u$ left-divides $\widehat\iota\cdot g(\psi_X, \tau_L)$.
  Said otherwise, $\Xi$ sends the ideal $\Hom(\psi,\phi)\iota\subset\End(\phi)$
  to the ideal $\iota\Hom(\psi,\phi)\subset\End(\psi)$.
  By considering the isomorphism $\Xi_{1, 2}:\End(\phi_1)\to\End(\phi_2)$ and by
  using the commutativity of $\End(\phi_2)$, we obtain
  \[\begin{array}{rcl}
      \Hom(\phi_3,\phi_2)\iota_2\cdot\Xi_{1,2}(\Hom(\phi_2,\phi_1)\iota_1)  &=&
      (\Hom(\phi_3,\phi_2)\iota_2)\cdot(\iota_1\Hom(\phi_2,\phi_1))\\
      &=& (\iota_1\Hom(\phi_2,\phi_1))\cdot (\Hom(\phi_3,\phi_2)\iota_2)\\
      &=& \Xi_{1,2}(\Hom(\phi_3,\phi_2)\Hom(\phi_2,\phi_1)\iota_2\iota_1)\\
      &\subset& \Xi_{1,2}(\Hom(\phi_3,\phi_1)\iota_2\iota_1).
  \end{array}\]

  To conclude, we use the properties of the norm of isogenies: the norm is
  multiplicative~\cite[Lem.~3.10.(i)]{gek91} and it corresponds to the norm of
  the associated ideal in $\A_\HH$~\cite[Lem.~3.10.(iv)]{gek91}. Consequently,
  the norms on both sides of the inclusion are equal. This implies that the last
  inclusion is in fact an equality.
\end{proof}

Algorithm~\ref{algo:isogeny-to-ideal} (\textsc{IsogenyToIdeal}) computes prime factors of the ideal in $\A_\HH$
corresponding to the given isogeny (of norm coprime to $\mathfrak p$), in order to recover the full factorization.
Each prime non-principal factor is treated independently by the subroutine
\textsc{PrimeIsogenyToPrimeIdeal}
(Algorithm~\ref{algo:prime-isogeny-to-prime-ideal}).

\begin{algorithm}\label{algo:prime-isogeny-to-prime-ideal}
	\caption{\PrimeIsogenyToPrimeIdeal}
  \KwIn{
    \begin{itemize}
      \item An ordinary Drinfeld module $\phi\in \Dr_2(\Fq[X], L)_\xi$,
      \item A monic prime $r\in\Fq[X]$ such that $r\notin\mathfrak p$,
      \item An $r$-isogeny $\iota:\phi\to\psi$ between
	$\phi,\psi\in\Dr_2(\Fq[X],
        L)_\xi$.
    \end{itemize}
    \vspace{-\topsep}
        }
	\KwOut{A polynomial $v\in\Fq[X]$ such that the left-ideal $\langle
	  \phi_{r}, \tau_L - \phi_{v}\rangle\subset
	\Ltau$ is generated by $\iota$.}
  $y \gets$ remainder in the right-division of $\tau_L$ by $\iota$\; \label{step:primeisogeny:1}
  $\iota^{(0)} \gets 1$ \;
  \For {$1 \leqslant n \leqslant \deg(r)$}{
    $\iota^{(n+1)} \gets$ remainder in the right-division $\phi_T \cdot \iota^{(n)}$ 
    by $\iota$ \;
  }
  using linear algebra, find $(v_0, \dots, v_{\deg(r)-1}) \in
  \Fq^{\deg(r)}$ such that
  $y - (v_0 \iota^{(0)} + \dots + v_{\deg(r)-1} \iota^{(\deg(r)-1)}) = 0$\label{step:v_def}\; 
  \KwRet $v_0 + v_1 X +\dots + v_{\deg(r)-1} X^{\deg(r)-1}$.
\end{algorithm}

\begin{algorithm}\label{algo:isogeny-to-ideal}
	\caption{\IsogenyToIdeal}
  \KwIn{
    \begin{itemize}
      \item An ordinary Drinfeld module $\phi\in \Dr_2(\Fq[X], L)_\xi$,
      \item A (non-necessarily prime) monic polynomial $u\in\Fq[X]$, such that
        $u\notin\mathfrak p$,
      \item A $u$-isogeny $\iota:\phi\to\psi$ between ordinary Drinfeld
	modules in $\Dr_2(\Fq[X], L)_\xi$,
    \end{itemize}
    \vspace{-\topsep}
        }
        \KwOut{A factorization of the ideal $\aid\subset \Fq[X,Y]/(\xi)$ associated to $\iota$ in
    Lemma~\ref{lem:corres_id_isog}.}
  \If{$u = 1$}{\KwRet $\Fq[X,Y]/(\xi)$.}
  $r\gets$ a nonconstant monic prime factor of $u$\; \label{step:isogenytoideal:factor}
  $\widetilde\iota \gets \rgcd(\iota, \phi_r)$\; \label{step:isogenytoideal:rgcd}
  \If{$\widetilde\iota = 1$}{\KwRet $\textsc{IsogenyToIdeal}(\phi,
  u/r^{\val_r(u)}, \iota)$.}
  \ElseIf{$\widetilde\iota = \lambda\phi_r$ for some $\lambda\in L^\times$}{
    \KwRet $\langle r(\overline X)\rangle\cdot\textsc{IsogenyToIdeal}(\phi,
    u/r, \iota\cdot \phi_r^{-1})$.}
  \Else{
      $v\gets \textsc{PrimeIsogenyToPrimeIdeal}(\phi, r, \widetilde\iota)$\;
      $\widetilde\phi\gets$ the codomain of $\widetilde\iota$,
      computed from $\phi$ and $\widetilde\iota$ with Formulas~\eqref{eq:Velu}\;
      \KwRet $\langle u(\overline X), \overline Y - v(\overline X)\rangle\cdot 
      \textsc{IsogenyToIdeal}(\widetilde\phi, u/r, \iota\cdot
      \widetilde\iota^{-1})$.
    }
\end{algorithm}

In what follows, $\theta$ is a feasible exponent for matrix multiplication in
$L$, satisfying
$2 \leqslant \theta \leqslant 3$.

Algorithm~\ref{algo:isogeny-to-ideal} involves the factorization of a polynomial $u \in
\Fq[X]$. We choose to use the Cantor-Zassenhaus algorithm
\cite{cantor-zassenhaus}, a Las Vegas probabilistic algorithm with expected complexity
bounded above by $\widetilde O(\delta^2 + \delta \log q)$, where $\delta$ is the
degree of the input. Another possibility was to use Berlekamp's algorithm,
which is deterministic. However, with a complexity dominated by
$\delta^\theta$, its use would severely hinder the overall complexity of the
algorithm. Finally, the complexities for
Algorithms~\ref{algo:prime-isogeny-to-prime-ideal} and
\ref{algo:isogeny-to-ideal} will be expressed in terms of $d$, $q$, and the
degree of the input polynomial $r$ (resp. $u$). As $\iota$ is an $r$-isogeny
(resp. $u$), its degree is bounded by that of $r$ (resp. $u$).

Before proving the correctness of
Algorithm~\ref{algo:prime-isogeny-to-prime-ideal}, we need the following
technical lemma:
\begin{lemma}\label{lem:separable-exists}
	If there exists an isogeny of norm $r\notin\mathfrak p$ between two finite
  Drinfeld $\A$-modules $\phi$ and $\psi$, then $\rgcd\left(\Hom(\psi,\phi)\right) = 1$.
\end{lemma}

\begin{proof}
  Let $f:\phi\to\psi$ be an $r$-isogeny, with $r\notin\mathfrak p$. Set $V = \bigcap_{u \in \Hom(\psi, \phi)}
	\Ker(u)$, and let $g$ be an isogeny in $\Hom(\psi, \phi)$. The sequence of
	$A$-modules $0 \to V \to \Ker(g) \to \Ker(g)/ V \to 0$ is exact, so that
  $\chi(V)$ divides $\chi(\Ker(g))$, where $\chi$ is the Euler-Poincaré
  characteristic, see Section~\ref{subsec:preli::finite}. Consequently,
	$\chi(V)\pid^{\height(g)/\deg(\pid)}\isonorm(f)$ divides $\isonorm(fg)$. In
	particular, $\chi(V)\isonorm(f)\mid\isonorm(fg)$. By
	\cite[Lem.~3.10.(iv)]{gek91}, we have $\sum_{g\in\Hom(\psi, \phi)}\isonorm(fg)
	= \isonorm(f)$. Since $\isonorm(f)\ne (0)$, $\chi(V)$ must equal $\A$ and
	hence $V=0$.


  Then $\Ker(\rgcd(\Hom(\psi, \phi))) = V$ is trivial, which implies
  that $\rgcd(\Hom(\psi, \phi))$ divides $\tau^{\deg(\pid)\ell}$ for some
  $\ell\in\N$. Since $r\notin\mathfrak p$, the $r$-dual $\hat f$ of $f$ is separable
  (it has norm $r\notin\mathfrak p$), hence $\rgcd\left(\Hom(\psi, \phi)\right)= 1$.
\end{proof}

\begin{proposition}\label{prop:prime-isogeny-to-prime-ideal}
  Algorithm~\ref{algo:prime-isogeny-to-prime-ideal}
  (\textsc{PrimeIsogenyToPrimeIdeal})
  is correct.
\end{proposition}

\begin{proof}
  First, we notice that since $r$ is prime, the norm of $\iota$ must be the
  ideal $(r)\subset\Fq[X]$, and hence $\deg_\tau(\iota) = \deg(r)$.
	Since $\iota$ is an $r$-isogeny, $\phi_r\in \Hom(\psi, \phi)\iota$. Since
	$\A_\HH$ is a Dedekind ring in a quadratic extension of $\Fq(X)$, the ideal $\Hom(\psi, \phi)\iota$ --- regarded as an ideal in $\A_\HH$ by
	Lemma~\ref{lem:corres_id_isog} --- contains the prime $r$. Therefore, it can
	only be either the full ring $\A_\HH$, the principal ideal $\langle r\rangle$,
	or a prime ideal of degree $1$ above $\langle r\rangle$.

	By Lemma~\ref{lem:separable-exists}, the left-ideal in $\Ltau$ generated by
	elements in $\Hom(\psi, \phi)\iota$ equals $\Ltau\iota$, which is neither the
	full ring $\Ltau$, nor $\Ltau\phi_r$ since $\deg_\tau(\phi_r) =
	2\deg(r)>\deg(\iota)$. Consequently, using the correspondence in
	Lemma~\ref{lem:corres_id_isog}, $\Hom(\psi, \phi)\iota$ must be a degree-$1$
	prime ideal above the principal ideal associated to $r$. Said otherwise, the
	polynomial $Y^2+h(X)Y - f(X)$ factors over $(\Fq[X]/(r))[Y]$, and a prime
	ideal above $\langle r\rangle$ in $\A_\HH$ has the form $\langle r(\overline
	X), Y-v(\overline X)\rangle$, where $v\in\Fq[X]$ satisfies $\xi(\overline X,
	\overline v)=0$ in $\Fq[X]/(r)$. Note that up to reducing $v$ modulo $r$, we
	can assume that $\deg(v) < \deg(r)$; under this assumption, $v$ is uniquely
	defined.

	We now prove that the coefficients of $v$ satisfy the equality in
	Step~\ref{step:v_def}, so that it can indeed be computed via linear algebra.
	To this end, we need to prove that $\iota$ right-divides $\tau_L -
	\phi_v$. This is a direct consequence of the fact that the ideal $\Hom(\psi,
	\phi)\iota\subset\End(\phi)$ corresponds to the ideal $\langle r(\overline X),
	Y-v(\overline X)\rangle\subset \A_\HH$.
\end{proof}

Algorithm~\ref{algo:isogeny-to-ideal} needs as input a polynomial
$u\in\Fq[X]$ such that $\iota$ right-divides $\phi_u$. It can be
found by looking for a non-trivial $\Fq$-linear relation between the remainders
of $\phi_{X^0},\phi_{X^1},\ldots,\phi_{X^\ell}$ in the right-division by
$\iota$. When $\ell\geq \deg_{\tau}(\iota)$, such a non-trivial linear
combination exists.

\begin{proposition}\label{prop:isogeny-to-ideal}
	Algorithm~\ref{algo:isogeny-to-ideal} (\textsc{IsogenyToIdeal})
	terminates and is correct.
\end{proposition}

\begin{proof}
  The proof is done by induction on the degree of $u$. The termination comes
  from the fact that the degree of $u$ decreases in each recursive call.

  By Lemma~\ref{lem:corres_id_isog}, there is a uniquely defined ideal $\aid\subset
  \A_\HH$ corresponding to $\iota$. Since $\A_\HH$ is Dedekind
  (Lemma~\ref{lem:pic0-clB}), $\aid$ factors as a product of prime ideals. For
  $r\in\Fq[X]$ an irreducible polynomial, we let $\aid_r$ denote the product of all
  primes in the factorization of $\aid$ which contain $\overline r\in\A_\HH$.
  Consequently, since $\overline u\in \aid$, we have \[\aid = \prod_{\substack{r\text{
  prime}\\r\text{ divides }u}} \aid_r.\] Let $r$ be a prime factor of $u$. Then
  there are three possible cases, depending on whether $r$ is inert, splits, or
  ramifies in $\A_\HH$.

  If $r$ is inert, then $\aid_r = \langle \overline r\rangle^\ell$ for some
  $\ell\geq 0$. If $\ell = 0$ then $\aid_r = \A_\HH$. In this case, if $u\ne
  1$, then $\overline
  r\notin \aid$ and therefore $\rgcd(\iota, \phi_r)=1$. Consequently, $\overline r$
  is invertible in $\aid$, and therefore $\overline u/\overline r^{\val_r(u)}$
  belongs to $\aid$ and we can apply our induction hypothesis. If $\ell>0$, then
  $\overline r$ divides all elements in $\aid$. Therefore $\phi_r$ right-divides
  $\iota$ and hence $\widetilde\iota=\lambda\phi_r$ for some $\lambda\in
  L^\times$. Since $\phi_r$ is an endomorphism of $\phi$, $\iota\cdot
  \phi_r^{-1}$ is a well-defined isogeny between $\phi$ and $\psi$ and its
  corresponding ideal in $\A_\HH$ is $\{g : g\in\A_\HH\mid g\cdot \overline r\in
  \aid\}$. This ideal contains $\overline u/\overline r$, hence we can apply
  our induction hypothesis.

  If $r$ splits then the ideal $\langle \overline r\rangle\subset\A_\HH$ factors
  as a product $\pid_1\cdot \pid_2$ of two distinct prime ideals. Therefore, $\aid_r =
  \pid_1^\alpha\cdot \pid_2^\beta$ for some $\alpha,\beta\geq 0$. First, if both
  $\alpha$ and $\beta$ are nonzero, then $\aid_r = \langle \overline r\rangle\cdot
  \pid_1^{\alpha-1} \pid_2^{\beta-1}$. Consequently, $\iota$ is right-divisible by
  $\phi_r$, $\widetilde\iota = \lambda\phi_r$ for some $\lambda\in L^\times$ and
  we can apply our induction hypothesis on the isogeny $\iota\cdot
  \phi_r^{-1}$. Now, we study the case where either $\alpha$ or $\beta$ is zero.
  Without loss of generality, let us assume that $\beta = 0$. Then $\aid_r =
  \pid_1^\alpha$. In this case, $\widetilde\iota$ cannot be right-divisible by
  $\phi_r$: this would contradict the fact that $\langle \overline r\rangle$
  does not divide $\aid$. On the other hand, $\widetilde\iota$ cannot equal $1$
  since for any element $g\in \pid_1$, $g(\phi_X, \tau_L)$ must right-divide both
  $\phi_r$ and $\iota$. Since $\iota$ is an isogeny, $\Ker(\iota)$ is an
  $\Fq[X]$-submodule of $\Lbar$ (for the module law induced by $\phi$), and
  hence so is $\ker(\widetilde\iota)=\ker(\iota)\cap \ker(\phi_r)$.
  Consequently, $\widetilde\iota$ is an isogeny from $\phi$ to some other
  Drinfeld module $\phi'\in\Dr_2(\Fq[X], L)_\xi$. The Drinfeld module $\phi'$
  can be computed using Formulas~\eqref{eq:Velu}, and the ideal corresponding to
  this isogeny can be computed using
  Algorithm~\ref{algo:prime-isogeny-to-prime-ideal}, which is correct by
  Proposition~\ref{prop:prime-isogeny-to-prime-ideal}. To apply the induction
  hypothesis on $\ell$, it remains to prove that $\iota' \coloneqq
  \iota\cdot\widetilde\iota^{-1}$ defines an isogeny $\iota':\phi'\to\psi$ which
  right-divides $\phi'_{u/r}$.
  To this end, let $\iota_{\rm dual}$ denote the dual $u$-isogeny of $\iota$,
  and let $\widetilde\iota_{\rm dual}$ be the dual $r$-isogeny of
  $\widetilde\iota$. We have \[\begin{array}{rcccccl} \phi'_u\phi'_r &=&
      \widetilde\iota\cdot \widetilde\iota_{\rm dual}\cdot\phi'_{u} &=&
      \widetilde\iota\cdot\phi_{u}\cdot \widetilde\iota_{\rm dual} &=&
      \widetilde\iota\cdot\iota_{\rm dual}\cdot\iota\cdot \widetilde\iota_{\rm
      dual}\\ &=& \widetilde\iota\cdot\iota_{\rm
      dual}\cdot\iota'\cdot\widetilde\iota\cdot \widetilde\iota_{\rm dual} &=&
    \widetilde\iota\cdot\iota_{\rm dual}\cdot\iota'\cdot\phi'_r.&& \end{array}\]
  By dividing on the right by $\phi'_r$, we obtain that $\iota'$ divides
  $\phi_u$ and that it is the $u$-dual of the composed isogeny
  $\widetilde\iota\cdot\iota_{\rm dual}$. This proves that $\iota'$ is a
  well-defined isogeny. By using the second statement in
  Lemma~\ref{lem:corres_id_isog}, we obtain that the ideal associated to
  $\iota'$ is \[\pid_1^{\alpha-1}\cdot \prod_{\substack{r'\text{ prime}\\r'\text{
  divides }u\\ r'\ne r}} \aid_{r'},\] which contains $\overline u/\overline r$,
  so that we can apply our induction hypothesis.

  Finally, the ramified case is proved similarly than the split case. The main
  difference is that $\pid_1 = \pid_2$, so that $\aid_r = \langle r\rangle^\ell\cdot
  \pid_1^\alpha$, for some $\ell\geq 0$ and $\alpha\in\{0,1\}$; this does not
  change the proof.
\end{proof}

\begin{proposition}
\label{prop:complexite:primeisogeny}

  Let $m$ denote the degree of $r$.
  Algorithm~\ref{algo:prime-isogeny-to-prime-ideal} (\PrimeIsogenyToPrimeIdeal) requires $O(dm^\theta)$
  operations in $L$ and $O(dm +m^2)$ applications of the Frobenius
  endomorphism.



\end{proposition}

\begin{proof}

  Computing the first remainder costs $O(dm)$ operations in $L$, and $O
  (dm)$ applications of the Frobenius endomorphism. The other remainders
  are computed recursively. Knowing $\iota^{(n)}$, computing $\iota^{(n+1)} = \phi_T \cdot
  \iota^{(n)}$ requires $O(m)$ operations in $L$, and the same number of
  Frobenius applications. This Ore polynomial has degree at most
  $\deg_\tau(\iota) + 1$. By Lemma
  \ref{lem:compl-euclidean-division},computing this remainder
  requires $O(\deg_\tau(\iota)) = O(m)$ operations in $L$, and as much applications of
  the Frobenius. Consequently, computing all elements in the loop requires
  $O(m^2)$ operations in $L$ and $O(m^2)$ applications of the Frobenius
  endomorphism.


  The last costly step is solving a linear system. More precisely, the
  algorithm finds a solution of an affine system over $\Fq$, whose associated
  matrix has less than $dm$ rows, and $m$ columns. Solving such a system
  requires $O(d m^\theta)$ operations in $L$. In total, we get $O(dm^\theta)$
  operations in $L$, and $O(dm + m^2)$ applications of the
  Frobenius endomorphism.

\end{proof}

\begin{proposition}
  \label{prop:complexite:isogeny}

  Let $m$ denote the degree of $u$. Using the Cantor-Zassenhaus algorithm for
  polynomial factorization, Algorithm~\ref{algo:isogeny-to-ideal}
  (\IsogenyToIdeal) is a
  probabilistic Las Vegas algorithm requiring $\widetilde O(d m^\theta+  m^3 +
  m \log(q))$ expected operations in $L$ and $O(dm+m^3)$ expected applications of the
  Frobenius endomorphism.


\end{proposition}

\begin{proof}
  Step~\ref{step:isogenytoideal:factor} is performed using the
  Cantor-Zassenhaus algorithm, with expected cost bounded by $\widetilde O(m^2 +
  m\log(q))$. Notice also that
  the initial factorization of $u$ may be performed only once for this cost, at the first
  call of the algorithm.  Then Step~\ref{step:isogenytoideal:rgcd} performs an Ore Euclidean division,
  which costs $O(m^2)$
  operations in $L$ and $O(m^3)$ applications of the Frobenius endomorphism
  using Euclid's algorithm (Lemma~\ref{lem:complexite-euclide-rgcd}).


  If $\widetilde\iota$ is $1$ or $\lambda \phi_r$, we just need to compute a polynomial
  division, and $\iota\cdot\phi_r^{-1}$ in the latter case. These
  computations do not exceed the complexity of
  Step~\ref{step:isogenytoideal:rgcd}. No other computation is performed and
  the algorithm is recursively called on a smaller instance. If $\widetilde\iota$ is neither $1$ or $\lambda \phi_r$, then
  Algorithm~\ref{algo:prime-isogeny-to-prime-ideal} is called. Let $u = r_1 \cdots
  r_\ell$ be a factorization of $u$ ($r_i$ prime, not necessarily distinct),
  and let $k_i:=\deg(r_i)$. We can
  assume that the prime factors are ordered as the algorithm
  processes them. Then, counting all the recursive calls of the algorithm, we
  get that the total expected cost is bounded above by
  $$ \begin{array}{c}\displaystyle\widetilde O(m^2 +  m\log(q)) + \sum_{i=1}^{\ell}
    \left(O\left(\sum_{j=i}^{\ell} k_j^2\right) + O(d\,k_i^\theta)\right)
    \text{\quad operations in $L$,}\\
    \displaystyle\sum_{i=1}^{\ell}\left(O\left(\sum_{j=i}^{\ell}
    k_j^2\right)+ O(dk_i+k_i^2) \right)\text{\quad applications of
    the Frobenius endomorphism.}
\end{array}$$
Since $\sum_{i=1}^\ell k_i = m$, we obtain that these formulas are bounded
above by 
  $\widetilde O(d m^\theta+  m^3 + m \log(q))$ expected operations in
  $L$ and $O(dm+m^3)$ expected applications of the Frobenius endomorphism.

\end{proof}

\begin{remark}
\label{rem:caruso-le_borgne}

  One could ask whether the complexities of Propositions~\ref{prop:complexity:groupaction}, 
  \ref{prop:complexite:primeisogeny} and
  \ref{prop:complexite:isogeny} may be enhanced by using more efficient
  algorithmic primitives for the arithmetic of Ore polynomials. In \cite{caruso-le_borgne-fact}
  and \cite{caruso-le_borgne-mult}, Caruso and Le Borgne provide new algorithms
  for Ore polynomial multiplication, Euclidean division and right-greatest
  common divisor; in many applications, those algorithms yield substantial
  speed-ups. We highlight that the authors work under the hypothesis that $q$
  is fixed and that operations in $L$ as well as applications of the Frobenius
  and $L$ cost $\widetilde O(d)$ operations in $\Fq$.

  Let us first ask ourselves if we can enhance Algorithm~\ref{algo:groupaction}
  by using asymptotically fast Ore Euclidean division at Step~\ref{step:iota}.
  Per \cite[Prop.~3.1]{caruso-le_borgne-mult}, computing $\iota$ would cost
  $\widetilde O (\SMgeq(d, d))$ operations in $\Fq$, where $\SMgeq$ is a
  function introduced in \cite[Sec.~3]{caruso-le_borgne-mult}. Using the values provided by
  the authors\footnote{Private discussions 
  with the authors revealed that the critical
  exponent $\frac{5 - \theta}{2}$ was mistyped in \cite{caruso-le_borgne-mult}, and should instead be
$\frac{2}{5 - \theta}$.}, we get $\SMgeq(d, d) = d^{\frac{9 - \theta}{5 -
\theta}}$. Even if we used $\theta = 2$, we would get $\frac{9 - \theta}{5 -
\theta} = \frac 7 3$, and the
  computation of $\iota$ would then be outweighed by the computations of
  $\tilde u$ and $\tilde v$, which both cost $\Otilde(d^3)$ operations in $\Fq$ in the complexity model of
  \emph{loc.~cit}. Therefore,
  using the algorithmic primitives of \cite{caruso-le_borgne-fact} and
  \cite{caruso-le_borgne-mult} does not at the moment improve the complexity
  bound
  in Proposition~\ref{prop:complexity:groupaction}. To benefit from
  those, one would need to enough reduce the cost of computing $\tilde u$ and $\tilde
  v$. Our attempts to do so were unsuccessful.

  The situation for Algorithm~\ref{algo:prime-isogeny-to-prime-ideal} is quite
  similar. The first step requires computing the remainder in the Euclidean
  division of $\tau_L$ by $\iota$. As $\iota$ has degree $O(m)$ and $\tau_L$
  has degree $d$, the computation would require $\widetilde O (\SMgeq(d + m, d))$
  operations in $\Fq$. With the formula for $\SMgeq$ in \cite[Sec.~3]{caruso-le_borgne-mult}, this is
  $\widetilde O ((d+m)d^{\frac{4}{5 - \theta}})$ operations in $\Fq$. Adding
  to that the $O(d^2m^\theta)$ operations in $\Fq$ required to solve the
  system, there is no benefit in using the algorithms of
  \cite{caruso-le_borgne-mult}. In fact, doing so would actually worsen the
  asymptotic complexity with respect to the variable $d$. This is due to the
  fact that the bound is linear with respect to $d
  + m$ for fixed $d$, but it has a costly dependence with respect to $d$.

  The complexity of Algorithm~\ref{algo:isogeny-to-ideal} depends on that of
  Algorithm~\ref{algo:prime-isogeny-to-prime-ideal}, and our conclusion is the
  same.

\end{remark}

\subsection{An explicit computation}

To demonstrate the practical effectivity of Algorithm~\ref{algo:groupaction}
(\GroupAction), we have implemented the group action for a hyperelliptic curve
of genus $260$ defined over $\mathbb F_2$. 

Our C++/NTL code is available at
\url{https://gitlab.inria.fr/pspaenle/crs-drinfeld-521}.

Set $L = \mathbb F_2[X] / \mathfrak p$,
where $\pid$ is the ideal generated by $X^{521}+X^{32}+1\in\mathbb
F_2[X]$. We encode polynomials using the hexadecimal NTL
notation: for
instance, $\texttt{0x4bc}$ denotes $X^2 + X^4 + X^5 + X^7 + X^{10} +
X^{11}\in\mathbb F_2[X]$. By extension, we also denote elements in $L$ by the NTL
hexadecimal convention, implicitly using the reduction modulo the ideal $\pid$.
Our isomorphism class of Drinfeld modules has $\j$-invariant (in $L$)
$$j_0 =\begin{array}{l}
\texttt{0xb985b4ce23bd9cf992f1176e17c27dab7ae67270131}\\
\texttt{\phantom{0x}12a2804cb64abccc7cce061e12786bb3248809922da}\\
\texttt{\phantom{0x}35d3b624d67d08087e07c260fcaa9807a420ca83fa95}.\end{array}$$

The coefficients of the characteristic
polynomial of the Frobenius endomorphism of the Drinfeld module
$\phi\in\Dr_2(\F_2[X], L)$ defined by $\phi_X = j_0^{-1}\tau^2 +\tau+\omega$ are:
\[\begin{cases}
  h =
  \begin{array}{l}\texttt{0xb1ffea4ab7e58b96adf4e4972d7db918}\\
    \texttt{\phantom{0x}4821c1d64b375df52669c60973bb80dee}
  \end{array}\in
  \mathbb F_2[X],\\
  f = X^{521}+X^{32}+1\in\mathbb F_2[X].
\end{cases}\]

The polynomial $Y^2 + h(X) Y - f(X)$ defines a genus-$260$ hyperelliptic curve $\HH$ over
$\F_2$, whose Picard group $\Pic^0(\HH)$ is cyclic and has almost-prime order
\scriptsize
\[2\times
315413182467545672604116316415047743350494962889744865259442943656024073295689.\]
\normalsize

This group order was computed using the Magma implementation of the
Denef-Kedlaya-Vercauteren algorithm
\cite{kedlaya2001counting, denef2006extension}. This computation costs 53
hours on a Intel(R) Xeon(R) CPU E7-4850.

We ran experiments for computing the group action on a laptop (Intel
  i5-8365U@1.60GHz CPU, $8$ cores,
16\,GB RAM). We chose an element of $\Pic^0(\HH)$ at random such that the
$u$-polynomial in the Mumford coordinates is irreducible and has degree $35$. The most costly step in practice is the first step of Euclid's
algorithm: it starts by computing $\tau^{521}$ modulo $\phi_u$, which has
$\tau$-degree $70$. Unfortunately, in our non-commutative setting we cannot use
binary exponentiation to speed-up this step: $(P_1 + \Ltau Q)\cdot (P_2+\Ltau
Q)$ need not equal $P_1\, P_2 + \Ltau Q$ for $P_1, P_2,Q \in\Ltau$. Therefore, we implemented a
parallelized
subroutine specialized for this task.
By
using the $8$ cores of the laptop, computing this group action takes 24\,ms. 

\bibliographystyle{abbrv}
\bibliography{bi}

\begin{thebibliography}{10}

\bibitem{bae1992singular}
S.~Bae and J.~K. Koo.
\newblock On the singular {D}rinfeld modules of rank 2.
\newblock {\em Mathematische Zeitschrift}, 210(1):267--275, 1992.

\bibitem{bombar}
M.~Bombar, A.~Couvreur, and T.~Debris-Alazard.
\newblock On codes and {L}earning {W}ith {E}rrors over function fields.
\newblock In {\em Advances in Cryptology -- CRYPTO 2022}, pages 513--540.
  Springer, 2022.

\bibitem{cantor-zassenhaus}
D.~G. Cantor and H.~Zassenhaus.
\newblock A new algorithm for factoring polynomials over finite fields.
\newblock {\em Mathematics of computation}, 36(154):587--592, 1981.

\bibitem{caranay-thesis}
P.~Caranay.
\newblock {\em Computing Isogeny Volcanoes of Rank Two {D}rinfeld Modules}.
\newblock PhD thesis, University of Calgary, 2018.

\bibitem{caranay-article}
P.~Caranay, M.~Greenberg, and R.~Scheidler.
\newblock Computing modular polynomials and isogenies of rank two {D}rinfeld
  modules over finite fields.
\newblock {\em Contemporary Mathematics}, 754:293--314, 2020.

\bibitem{caruso-le_borgne-mult}
X.~Caruso and J.~Le~Borgne.
\newblock Fast multiplication for skew polynomials.
\newblock {\em Proceedings of the 2017 ACM on International Symposium on
  Symbolic and Algebraic Computation}, pages 77--84, 2017.

\bibitem{caruso-le_borgne-fact}
X.~Caruso and J.~Le~Borgne.
\newblock A new faster algorithm for factoring skew polynomials over finite
  fields.
\newblock {\em Journal of Symbolic Computation}, 79:411--443, 2017.

\bibitem{csidh}
W.~Castryck, T.~Lange, C.~Martindale, L.~Panny, and J.~Renes.
\newblock {CSIDH}: An efficient post-quantum commutative group action.
\newblock In {\em Asiacrypt 2018}, pages 395--427. Springer, 2018.

\bibitem{cohen2005handbook}
H.~Cohen, G.~Frey, R.~Avanzi, C.~Doche, T.~Lange, K.~Nguyen, and
  F.~Vercauteren.
\newblock {\em Handbook of Elliptic and Hyperelliptic Curve Cryptography}.
\newblock CRC, 2005.

\bibitem{couveignes2006hard}
J.-M. Couveignes.
\newblock Hard homogeneous spaces.
\newblock Cryptology ePrint Archive, Report 2006/291,
  \url{https://ia.cr/2006/291}, 2006.

\bibitem{cox_primes_2013}
D.~A. Cox.
\newblock {\em Primes of the form $x^2+ny^2$: {Fermat}, Class Field Theory, and
  Complex Multiplication}.
\newblock Wiley, 2nd edition, Apr. 2013.

\bibitem{deligne1987survey}
P.~Deligne and D.~Husemoller.
\newblock Survey of {D}rinfel'd modules.
\newblock {\em Contemporary mathematics}, 67:25--91, 1987.

\bibitem{denef2006extension}
J.~Denef and F.~Vercauteren.
\newblock An extension of {K}edlaya's algorithm to hyperelliptic curves in
  characteristic 2.
\newblock {\em Journal of cryptology}, 19(1):1--25, 2006.

\bibitem{doliskani2021drinfeld}
J.~Doliskani, A.~K. Narayanan, and {\'E}.~Schost.
\newblock Drinfeld modules with complex multiplication, {H}asse invariants and
  factoring polynomials over finite fields.
\newblock {\em Journal of Symbolic Computation}, 105:199--213, 2021.

\bibitem{drinfel1974elliptic}
V.~G. Drinfel'd.
\newblock Elliptic modules.
\newblock {\em Mathematics of the USSR-Sbornik}, 23(4):561--592, 1974.

\bibitem{dummit1994rank}
D.~Dummit and D.~Hayes.
\newblock Rank-one {D}rinfeld modules on elliptic curves.
\newblock {\em Mathematics of Computation}, 62(206):875--883, 1994.

\bibitem{garai2022computing}
S.~Garai and M.~Papikian.
\newblock Computing endomorphism rings and {F}robenius matrices of {D}rinfeld
  modules.
\newblock {\em Journal of Number Theory}, 237:145--164, 2022.

\bibitem{gekeler1983}
E.-U. Gekeler.
\newblock Zur arithmetik von {D}rinfeld-moduln.
\newblock {\em Mathematische Annalen}, 262(2):167--182, 1983.

\bibitem{gek91}
E.-U. Gekeler.
\newblock On finite {D}rinfeld modules.
\newblock {\em Journal of algebra}, 1(141):187--203, 1991.

\bibitem{gos98}
D.~Goss.
\newblock {\em Basic Structures of Function Field Arithmetic}.
\newblock Springer, 1998.

\bibitem{hayes}
D.~R. Hayes.
\newblock A brief introduction to {D}rinfeld modules.
\newblock In {\em The Arithmetic of Function Fields: Proceedings of the
  Workshop at the Ohio State University, June 17-26, 1991}, pages 1--32. De
  Gruyter, 1991.

\bibitem{sidh}
D.~Jao and L.~De~Feo.
\newblock Towards quantum-resistant cryptosystems from supersingular elliptic
  curve isogenies.
\newblock In {\em Post-Quantum Cryptography}, pages 19--34. Springer, 2011.

\bibitem{jn}
A.~Joux and A.~K. Narayanan.
\newblock Drinfeld modules may not be for isogeny based cryptography.
\newblock Cryptology ePrint Archive, \url{https://ia.cr/2019/1329}, 2019.

\bibitem{kedlaya2001counting}
K.~S. Kedlaya.
\newblock Counting points on hyperelliptic curves using {M}onsky-{W}ashnitzer
  cohomology.
\newblock {\em Journal of the Ramanujan Mathematical Society}, 16:323--338,
  2001.

\bibitem{kuehnert_isomorphic_2022}
B.~Kuehnert, G.~Schlafly, and Z.~Yi.
\newblock On isomorphic {K}-rational groups of isogenous elliptic curves over
  finite fields.
\newblock {\em Rose-Hulman Undergraduate Mathematics Journal}, 23(1), 2022.

\bibitem{kuhn2022finding}
N.~Kuhn and R.~Pink.
\newblock Finding endomorphisms of {D}rinfeld modules.
\newblock {\em Journal of Number Theory}, 232:118--154, 2022.

\bibitem{lang1987elliptic}
S.~Lang.
\newblock {\em Elliptic functions}.
\newblock Springer, 1987.

\bibitem{lang}
S.~Lang.
\newblock {\em Algebra}.
\newblock Springer, 3rd edition, 2002.

\bibitem{lorenzini1996invitation}
D.~Lorenzini.
\newblock {\em An Invitation to Arithmetic Geometry}.
\newblock American Mathematical Society, 1996.

\bibitem{schost}
Y.~Musleh and E.~Schost.
\newblock Computing the characteristic polynomial of a finite rank two
  {D}rinfeld module.
\newblock In {\em Proceedings of the 2019 ACM on International Symposium on
  Symbolic and Algebraic Computation}, pages 307--314. ACM, 2019.

\bibitem{ros02}
M.~Rosen.
\newblock {\em Number Theory in Function Fields}.
\newblock Springer, 2002.

\bibitem{rostovtsev2006public}
A.~Rostovtsev and A.~Stolbunov.
\newblock Public-key cryptosystem based on isogenies.
\newblock Cryptology ePrint Archive, \url{https://ia.cr/2006/145}, 2006.

\bibitem{silverman1994advanced}
J.~H. Silverman.
\newblock {\em Advanced topics in the arithmetic of elliptic curves}.
\newblock Springer, 1994.

\bibitem{villa-salvador}
G.~Villa~Salvador.
\newblock {\em Topics in the Theory of Algebraic Function Fields}.
\newblock Birkhäuser, 2006.

\bibitem{von_zur_gathen_modern_2013}
J.~von~zur Gathen and J.~Gerhard.
\newblock {\em Modern Computer Algebra}.
\newblock Cambridge University Press, 3 edition, 2013.

\bibitem{cryptoeprint:2022/438}
B.~Wesolowski.
\newblock Computing isogenies between finite {D}rinfeld modules.
\newblock Cryptology ePrint Archive, Paper 2022/438, 2022.
\newblock \url{https://eprint.iacr.org/2022/438}.

\end{thebibliography}
\end{document}